\documentclass{article}
\usepackage{graphicx} 
\def\withcolors{1}
\def\withnotes{1}
\def\anonymous{0}
\usepackage{ccanonne}
\usepackage{multicol}
\usepackage{multirow}
\usepackage{subcaption}

\allowdisplaybreaks
\newcolumntype{C}{>{\centering\arraybackslash}p{0.27\textwidth}}
\newcommand{\citep}{\cite}
\newcommand{\citet}{\cite}

\newcommand{\dims}{d}

\newcommand{\nspu}{m}

\newcommand{\unif}{{\mathbf{u}}}
\newcommand{\ngr}{{n_0}}

\newcommand{\rank}{{\text{rank}}}

\newcommand{\Proj}{{\Pi}}
\newcommand{\povmset}{{\mathfrak{M}}}
\newcommand{\nqubits}{{N}}
\newcommand{\pauliI}{{\sigma_I}}
\newcommand{\pauliX}{{\sigma_X}}
\newcommand{\pauliY}{{\sigma_Y}}
\newcommand{\pauliZ}{{\sigma_Z}}
\newcommand{\pauliObsSet}{{\mathcal{P}}}
\newcommand{\prPauli}[2]{{\p_{#2}(#1)}}

\newcommand{\opnorm}[1]{{\left\|#1\right\|}_{\text{op}}}
\newcommand{\tracenorm}[1]{{\left\|#1\right\|}_{1}}
\newcommand{\hsnorm}[1]{{\left\|#1\right\|}_{\text{HS}}}
\newcommand{\barDelta}{{\overline{\Delta}}}
\newcommand{\ptb}{{z}}
\newcommand{\ptbDistr}{{\mathcal{D}_{\cd}}}

\newcommand{\supparen}[1]{^{(#1)}}

\newcommand{\cd}{{c}}
\newcommand{\cop}{\kappa}
\newcommand{\simprob}{\eta}

\newcommand{\isthestate}{\texttt{YES}}
\newcommand{\notthestate}{\texttt{NO}}

\newcommand{\bfP}{\mathbf{P}}
\newcommand{\bfQ}{\mathbf{Q}}

\newcommand{\x}{\mathbf{x}}

\def\multiset#1#2{\ensuremath{\left(\kern-.3em\left(\genfrac{}{}{0pt}{}{#1}{#2}\right)\kern-.3em\right)}}

\newcommand{\qs}{\rho}
\newcommand{\qmm}{{\rho_{\text{mm}}}}
\newcommand{\qkn}{{\rho_0}}

\newcommand{\polylog}{{\text{polylog}}}

\newcommand{\HH}{\mathbb{H}}
\newcommand{\Herm}[1]{{\HH_{#1}}}

\newcommand{\qbit}[1]{|#1\rangle}
\newcommand{\qadjoint}[1]{\langle#1|}
\newcommand{\qproj}[1]{\qbit{#1}\qadjoint{#1}}
\newcommand{\qoutprod}[2]{\qbit{#1}\qadjoint{#2}}

\newcommand{\qdotprod}[2]{\langle#1|#2\rangle}
\newcommand{\hdotprod}[2]{\left\langle#1,#2\right\rangle}
\newcommand{\matdotprod}[3]{\langle#1|#2|#3\rangle}

\newcommand{\eye}{\mathbb{I}}
\newcommand{\Real}{\text{Re}}
\newcommand{\Img}{\text{Im}}

\newcommand{\VecOp}{\text{vec}}
\newcommand{\vvec}[1]{|#1\rangle\rangle}
\newcommand{\vadj}[1]{\langle\langle#1|}
\newcommand{\vvdotprod}[2]{\langle\langle#1|#2\rangle\rangle}

\newcommand{\bx}{\mathbf{x}}

\newcommand{\Luders}{\mathcal{H}}
\newcommand{\avgLuders}{{\overline{\Luders}}}
\newcommand{\Choi}{{\mathcal{C}}}
\newcommand{\avgChoi}{{\overline{\Choi}}}
\newcommand{\hbasis}{{\mathcal{V}}}

\DeclareMathOperator{\diag}{diag}

\newcommand{\PiRank}{r}
\newcommand{\Wg}{\text{Wg}}

\newcommand{\Sim}{\mathcal{S}}
\newcommand{\Mob}{\text{Mob}}
\newcommand{\Cat}{\text{Cat}}
\newcommand{\Haar}[1]{{\mathcal{U}_{#1}}}
\newcommand{\POVM}{\mathcal{M}}
\newcommand{\permProd}[2]{{\left\langle#1\right\rangle_{#2}}}
\newcommand{\cycle}{\mathcal{C}}

\let\svthefootnote\thefootnote
\newcommand\blankfootnote[1]{%
  \let\thefootnote\relax\footnotetext{#1}%
  \let\thefootnote\svthefootnote%
}

\title{Quantum state testing with restricted measurements}
\ifnum\anonymous=1
    \author{Anonymous authors}
\else
    \author{\begin{tabular}{C C }
  Yuhan Liu  & Jayadev Acharya \\
 Cornell University & Cornell University \\ 
\small \texttt{yl2976@cornell.edu} &\small \texttt{acharya@cornell.edu} 
\end{tabular}
}
\fi

\begin{document}

\maketitle
\begin{abstract}
    We study quantum state testing where the goal is to test whether $\rho=\rho_0\in\mathbb{C}^{d\times d}$ or $\|\rho-\rho_0\|_1>\varepsilon$, given $n$ copies of $\rho$ and a known state description $\rho_0$. In practice, not all measurements can be easily applied, even using unentangled measurements where each copy is measured separately. We develop an information-theoretic framework that yields unified copy complexity lower bounds for restricted families of non-adaptive measurements through a novel \emph{measurement information channel}. 
Using this framework, we obtain the optimal bounds for a natural family of $k$-outcome measurements with fixed and randomized schemes. We demonstrate a separation between these two schemes, showing the power of randomized measurement schemes over fixed ones.
Previously, little was known for fixed schemes, and tight bounds were only known for randomized schemes with $k\ge d$ and  Pauli observables, a special class of 2-outcome measurements. Our work bridges this gap in the literature.  
\end{abstract}
\blankfootnote{A part of this work was published in the Proceedings of Conference on Learning Theory 2024. This work was supported by NSF award 1846300 (CAREER), NSF CCF-1815893, and a research scholar grant from Google. }

\newpage
\tableofcontents
\newpage

\section{Introduction}

We consider the problem of quantum state testing/certification~\cite{ODonnellW15,wright2016learn,BadescuO019}, where the objective is to verify whether the output of a quantum system is a desired quantum state. It is a special case of quantum property testing where the goal is to test if a quantum state has a certain property~\cite{Montanaro2016}. Property testing \cite{goldreich2017introduction, canonne2020survey} in general has gained significant attention in information theory and theoretical computer science. 
Specifically, given $\ns$ copies of an unknown state $\rho\in\C^{\dims\times\dims}$ and a known state description of $\qkn$, 
we wish to use quantum measurements to determine whether $\rho=\qkn$ or $\tracenorm{\rho-\qkn}>\eps$, where $\tracenorm{\cdot}$ is the trace norm\footnote{Previous works also considered other measures such as fidelity and Bures $\chi^2$-divergence, e.g. \cite{BadescuO019}.}. 
When $\qkn=\qmm\eqdef\eye_\dims/\dims$, the maximally mixed state, the problem is called \emph{mixedness testing}.


Our goal is to characterize the \emph{copy complexity}, the minimum number of copies of unknown states to solve the testing problem. Prior work of~\cite{ODonnellW15, BadescuO019} showed that $\ns=\Theta(\dims/\eps^2)$ copies are necessary and sufficient. The dimension of quantum states is $\sim\dims^2$, and thus the copy complexity is sublinear in the dimension. However, their result relies on \emph{entangled measurements}, where we treat the $\ns$ copies of $\rho$ as one giant state $\rho^{\otimes \ns}$ and apply arbitrary measurements to it. This is extremely difficult in practice even for moderate values of $\dims$ and $\ns$ as it requires a quantum computer with a lot of qubits that can sustain coherence for a long time. 

Thus, it is reasonable to study the problem under \emph{measurement restrictions}, namely measurements that are easier to implement than an entangled measurement. Naturally, weaker measurement means more copies are needed, and the copy complexity may not be sublinear anymore. This leaves an important question,

\begin{center}
    \fbox{Given a measurement restriction, can sublinear copy complexity be achieved for state certification?}
\end{center}

Several works have considered \emph{unentangled measurements}\footnote{Also called \emph{incoherent} and \emph{independent} measurements in previous literature} where measurements are performed on one copy of $\rho$ at a time. They are easier to implement than entangled measurements, and thus will be the focus of our work. There are three types of unentangled \emph{measurement schemes} which determine how the measurement for each copy is chosen: \textbf{fixed, randomized}, and \textbf{adaptive}, in the order of increasing generality and power, which can lead to lower copy complexity at the cost of increasing difficulty in implementation. We introduce the formal definitions of these measurement schemes in \cref{sec:setup}.

When \emph{no further measurement restrictions} are present, previous works have established the optimal copy complexity for quantum state certification for the three types of unentangled measurements. For randomized non-adaptive and adaptive schemes, the optimal copy complexity is $\Theta(\dims^{3/2}/\eps^2)$~\cite{BubeckC020,Chen0HL22}. However, for deterministic schemes, the optimal copy complexity is $\ns=\Theta(\dims^2/\eps^2)$ \cite{Yu21sample, liu2024role}. In short, randomness is necessary to achieve sublinear copy complexity.

The unentangled measurement schemes proposed in the previous works allow for arbitrary measurements to be performed on each copy. 
However, we still may not be able to apply arbitrary measurements to each copy and can only choose each measurement from a subset $\povmset$. 
Due to physical restrictions, the measurement instrument may only have a finite number of outcomes. For example, a single-photon detector may only detect a finite number of photons in a fixed timeframe due to limited temporal resolution. 
Another notable example is the \emph{Pauli observables}, a class of 2-outcome measurements associated with the Pauli operators which are simple to implement on quantum computers. Thus, it is practically reasonable to consider measurements with only a limited number of outcomes. 

In this work, we study quantum state certification with restricted unentangled measurements, particularly focusing on the case when the number of potential outcomes of measurements is limited to at most $\ab$. This is an important step in understanding quantum property testing with restricted measurements. We extend the lower bound techniques in \cite{liu2024role} to restricted measurements and design optimal algorithms for finite outcome measurements.

\subsection{Problem setup}
\label{sec:setup}

\subsubsection{Quantum states and measurements} We first introduce notations and basics of quantum computing.
We use the Dirac notation $\qbit{\psi}$ to denote a vector in $\C^{\dims}$. $\qadjoint{\psi}\eqdef(\qbit{\psi})^\dagger$ is the conjugate transpose, which is a row vector. $\qdotprod{\psi}{\phi}$ is the Hibert-Schmidt inner product of $\qbit{\psi}$ and $\qbit{\phi}$. We denote the set of all $\dims\times\dims$ Hermitian matrices by $\Herm{\dims}$. A $\dims$-dimensional quantum system is described by a positive-semidefinite Hermitian matrix $\rho\in\Herm{\dims}$ with $\Tr[\rho]=1$. We assume $\dims=2^{\nqubits}$ where $\nqubits$ is the number of qubits in the system.

Measurements are formulated as \emph{positive operator-valued measure} (POVM). Let $\mathcal{X}$ be an outcome set. Then a POVM $\POVM=\{M_x\}_{x\in \mathcal{X}}$, where $M_x$ is p.s.d. and $\sum_{x\in \mathcal{X}}M_x=\eye_\dims$. Let $X$ be the outcome of measuring $\rho$ with $\POVM$, then the probability observing $x\in\mathcal{X}$ is given by the \emph{Born's rule},
\[
\probaOf{X=x}=\Tr[\rho M_x].
\]
For measurements with at most $\ab$ outcomes, $\mathcal{X}=[\ab]\eqdef\{1, \ldots, \ab\}$ is finite\footnote{The set $\mathcal{X}$ can be countably or uncountably infinite. In both cases, POVMs and Born's rule can be generalized.}, or equivalently $\log_2\ab$ bits of classical memory.

\subsubsection{Quantum state testing with unentangled measurements}
Given $\ns$ copies of an unknown state $\rho\in\Herm{\dims}$, we wish to design
\begin{itemize}
    \item $\ns$ POVMs $\POVM^\ns=(\POVM_1, \ldots, \POVM_\ns)$, where $\POVM_i=\{M_x\supparen{i}\}$ is applied to the $i$th copy and produce an outcome $x_i$, which follow a discrete distribution $\p_{\rho}\supparen{i}=[\p_{\rho}\supparen{i}(1), \ldots, \p_{\rho}\supparen{i}(\ab)]$ where $\p_{\rho}\supparen{i}(x)=\Tr[M_x\supparen{i}\rho]$ is defined by the Born's rule. Define $\bx=(x_1, \ldots, x_{\ns})\in[\ab]^\ns$ to be the collection of all outomces.
    \item A tester $T:[\ab]^\ns\mapsto\{\isthestate, \notthestate\}$ that predicts whether $\rho=\qkn$ or $\tracenorm{\rho-\qkn}>\eps$ based on the outcomes $\bx$. We wish to guarantee that the answer is correct with probability at least 2/3, 
    \[
\Pr_{\rho = \qkn}(T(\bx) = \isthestate) \ge \frac23,\quad  \text{and} \inf_{\rho:\tracenorm{\rho-\qkn}>\eps}\Pr(T(\bx) = \notthestate) \ge \frac 23.
\]
\end{itemize}

When $\qkn=\qmm\eqdef\eye_\dims/\dims$, the problem is called \emph{mixedness testing}. We wish to characterize the \emph{copy complexity}, the minimum $\ns$ such that there exists a tester and measurement scheme that achieves the desired $2/3$ success probability for all $\qkn$.

Given $\POVM^\ns$, we define $\bfP_{\rho}$ as the distribution of all outcomes $\bx$ when the state is $\rho$.
We formulate the three unentangled measurement schemes and their outcome distributions $\bfP_{\rho}$ and discuss their advantages and drawbacks.
\begin{description}
\item{\textbf{Fixed.}} $\POVM_i$'s are determined \emph{before} receiving copies of $\rho$, and $\bfP_\rho=\otimes_{i=1}^{\ns}\p_\rho\supparen{i}$ is a product distribution. 

A key advantage of such protocols is that the same set of measurements can be used for multiple repetitions of the testing problem. Moreover, there is no latency since the measurements are not designed after the states are made available, which is a drawback of the following protocols.

\item{\textbf{Randomized non-adaptive.}} There is a random seed $R\sim\mathcal{R}$, and each measurement $\POVM_i=\POVM_i(R)$ is a function of $R$. \emph{Conditioned on $R$}$, \bfP_\rho$ is a product distribution.

Compared to fixed measurements, the drawback is that whenever we wish to run the task again, 
we need to select a new set of measurements using a new instance of the common randomness\footnote{If the set of measurements is finite, we can prepare all measurements beforehand and sample with classical randomness. However, this could still be difficult if the set is very large.}. 
Preparing the new measurements can be costly, and in some cases, the random sampling process may not even be feasible.

\item{\textbf{Adaptive.}} The $i$th measurement depends on a random seed $R\sim\mathcal{R}$ and all previous outcomes, $\POVM_i=\POVM_i(x_1, \ldots, x_{i-1}, R)$. In general, $\bfP_{\rho}$ is not a product distribution.

A primary drawback of this scheme is the latency and complications associated with designing measurements one after another. We note that since the first measurement outcome can be used as a source of common randomness, adaptive schemes are inherently randomized.

\end{description}

\subsection{Prior results}
When there is no other constraint on measurements besides the measurement scheme, the copy complexity of quantum state certification is well understood.
For entangled measurements, \cite{ODonnellW15} showed that $\ns=\Theta(\dims/\eps^2)$ is necessary and sufficient for mixedness testing. \cite{BadescuO019} further showed that $\ns=O(\dims/\eps^2)$ is sufficient for testing against all states $\qkn$.
\cite{BubeckC020, ChenLO22instance} showed that for randomized non-adaptive measurements, the copy complexity is $\Theta(\dims^{3/2}/\eps^2)$. \cite{Chen0HL22} further showed that adaptivity does not help improve the copy complexity. The role of randomness in quantum state certification was very recently initiated. \cite{Yu21sample} showed that $\ns=O(\dims^2/\eps^2)$ is sufficient. A recent work by~\cite{liu2024role} showed that $\ns=\Omega(\dims^2/\eps^2)$ is also necessary, thus demonstrating a separation with randomized measurements. 

To our knowledge, the problem with measurement restrictions beyond measurement schemes has not been extensively studied. We are only aware of \cite{Yu2023almost} which showed that with randomized and non-adaptive Pauli observables, $\ns={\Theta}(\dims^2\polylog(\dims)/\eps^2)$ copies are necessary and sufficient for state certification. However, the lower bound only holds for Pauli observables and is not general enough even for arbitrary two-outcome measurements. Thus, a huge mystery remains for measurements with a finite number of outcomes in the simplest non-adaptive setting, where we do not know whether the copy complexity can be sublinear, let alone the exact dependency on the number of outcomes. It is also unknown whether randomness helps, even for the simplest case of Pauli observables. 

\subsection{New results}
\label{sec:new-results}
We prove copy complexity bounds for both randomized and fixed measurements for the problem of unentangled quantum state certification when each measurement is restricted to have at most $\ab$ outcomes. 

\paragraph{Randomized measurements.}
We completely characterize the copy complexity for randomized non-adaptive finite-outcome measurements in~\cref{thm:random-k}.
\begin{theorem}
\label{thm:random-k}
    With randomized non-adaptive $k$-outcome measurements where $\ab\ge 2$,  
    \[
    \ns=\bigTheta{\frac{\dims^2}{\eps^2\sqrt{\min\{\ab, \dims\}}}}
    \]
    copies are necessary and sufficient to test whether $\rho=\qkn$ or $\tracenorm{\rho-\qkn}>\eps$ with probability at least $2/3$. 
\end{theorem}
Setting $\ab\ge\dims$ recovers the $\Theta(\dims^{3/2}/\eps^2)$ bound for general randomized schemes in \cite{BubeckC020}\footnote{The upper bound was corrected by \cite{ChenLO22instance}.}.

\paragraph{Fixed measurements.} We further investigate whether randomization is necessary to achieve the tight bounds above and prove a tight bound for fixed measurements in~\cref{thm:fixed-k-lower}.
\begin{theorem}
\label{thm:fixed-k-lower}
    For all $\ab\ge 2$,  with fixed unentangled 
    $k$-outcome measurements, 
    \[
    \ns=\bigTheta{\frac{\dims^3}{\eps^2{\min\{\ab, \dims\}}}}
    \]
    copies are necessary and sufficient to test whether $\rho=\qkn$ or $\tracenorm{\rho-\qkn}>\eps$ with probability at least $2/3$.
\end{theorem}
Setting $\ab\ge \dims$ recovers the $\Theta(\dims^2/\eps^2)$ bound for general fixed measurement schemes in \cite{Yu21sample,liu2024role}.

Observe a large separation between the bounds in \cref{thm:random-k,thm:fixed-k-lower}. 
To see it more clearly for each $\ab$, we set $k=d^\alpha$ where $\alpha\in(0, 1)$. With randomness, the copy complexity in \cref{thm:random-k} is
\[
\ns=\Theta(d^{2-\alpha/2}/\eps^2),
\]
which is sublinear (in the dimensionality of quantum states). Without randomness, \cref{thm:fixed-k-lower} leads to a bound of 
\[
\ns=\Theta(d^{3-\alpha}/\eps^2),
\]
which is superlinear. Thus, the copy complexity changes from sublinear to superlinear only due to the lack of randomness. The ratio between fixed and randomized copy complexity is $\Theta(\dims^{1-\alpha/2}/\eps^2)$, which is between $\sqrt{\dims}$ and $\dims$.
This result quantitatively demonstrates the power of randomness in quantum state testing.

For $\ab=2$ in \cref{thm:fixed-k-lower}, we design an optimal algorithm with fixed Pauli observables that uses $\ns=O(\dims^3/\eps^2)$ copies, thereby completely characterizing the copy complexity for fixed Pauli observables and 2-outcome measurements. 
\begin{theorem}
\label{thm:fixed-pauli-upper-lower}
    Using fixed Pauli observables or 2-outcome measurements, $\ns=\Theta(\dims^3/\eps^2)$ copies are necessary and sufficient to test whether $\rho=\qkn$ or $\tracenorm{\rho-\qkn}>\eps$ with probability at least $2/3$.
\end{theorem}
Thus there is a strict separation with the bound of $\ns=\tilde{\Theta}(\dims^2/\eps^2)$ for random Pauli observables in~\cite{Yu2023almost}. 

\subsection{Related work}

\paragraph{Quantum state certification} \cite{BadescuO019,ChenLO22instance, Yu21sample,Yu2023almost} proposed algorithms that also apply to closeness testing. \cite{ChenLO22instance,Chen0HL22} went beyond worst-case bounds and derived near-optimal bounds for unentangled measurements that decrease when $\qkn$ is approximately low rank. While instance-optimal bounds have very important practical implications, we note that for finite-outcome measurements even the worst-case bounds were not known. 

\paragraph{Related quantum state inference problems} \emph{Quantum tomography}~\citep{ODonnellW16, wright2016learn, ODonnellW17,guctua2020fast,flamian2023tomography} aims learn the full state description of the unknown state $\rho$. Tomography algorithms can be applied to state certification, though with far more samples than needed. For full rank $\rho$, to estimate up to $\eps$ in trace distance, $\Theta(\dims^3/\eps^2)$ copies are necessary and sufficient for unentangled measurements \citep{kueng2017low,HaahHJWY17}, even with adaptivity~\cite{chen2023does}. For non-adaptive Pauli observables, the bound is $\Theta(\dims^4/\eps^2)$~\citep{flamia2011direct,lowe2022lower}. In particular,~\cite {lowe2022lower} developed a lower bound technique for measurements with constant number of outcomes, but their result is not explicitly stated for general $\ab$.

There are other closely related problems. \cite{ogawa2000strong, brandao2020adversarial,regula2023postselected} studied hypothesis testing where the goal is to distinguish between two known states. Hypothesis selection~\citep{BadescuO21, fawzi:hal-04107265} aims to determine $\rho$ from a finite set of hypothesis sets. Shadow tomography~\citep{Aaronson20,huang2020predicting,BrandaoKLLSW19} aims to learn the statistic of $\rho$ over a finite set of observables.

\paragraph{Classical distribution testing}
Quantum state certification is a generalization of the classical problem of testing discrete distributions, where the goal is to decide whether samples from an unknown distribution are from a desired distribution. Entangled measurement is analogous to the centralized setting where all samples are accessible in one place, while unentangled measurement is similar to the distributed case where samples are over multiple devices and not directly accessible. 

Centralized testing of discrete distributions has been well-studied since~\cite{BatuFFKRW01, Paninski08}. Recently, various works studied information-constrained testing in the distributed setting.~\cite{duchi2013local, ACFT:19:IT3} focused on privacy-preserving information.~\cite{barnes2019lower, ACT:19:IT2} considered communication-constrained inference where each device can only send limited bits for each sample. This is in some sense the classical equivalent of quantum state certification using unentangled finite-outcome measurements. Our framework is qualitatively similar to~\cite{AcharyaCT19,ACLST22iiuic} which developed a unified lower bound framework for information-constrained testing and learning of distributions. We refer the readers to~\cite{canonne2022topics} for a survey of the above topics.

\paragraph{Organization.} In~\cref{sec:techniques} we summarize our key technical contributions. In~\cref{sec:preliminaries} we describe additional technical preliminaries. In~\cref{sec:MIC} we introduce important properties of the measurement information channel. In~\cref{sec:lower} we prove the key lower bound of~\cref{thm:lower-channel}. In~\cref{sec:rand-k-upper} we introduce the algorithm using randomized $\ab$-outcome measurements. We introduce the algorithms using fixed Pauli observables  in~\cref{sec:fixed-Pauli-upper} and general $\ab$-outcome measurements in \cref{sec:fixed-d-outcome}.

\section{Our techniques}
\label{sec:techniques}
\subsection{Lower bound via measurement information channel}
One of our main contributions is a unified lower bound framework for restricted unentangled measurements which relates the hardness of testing to a novel \emph{measurement information channel (MIC)}.

\begin{definition}
\label{def:mic}
    Let $\POVM=\{M_x\}_{x}$ be a POVM. The \emph{measurement information channel (MIC)} $\Luders_{\POVM}:\C^{\dims\times\dims}\mapsto\C^{\dims\times\dims}$ and its matrix representation $\Choi_{\POVM}$ are defined as
\begin{equation}
    \Luders_{\POVM}(A)\eqdef\sum_{x}M_x\frac{\Tr[M_xA]}{\Tr[M_x]}, \quad\Choi_{\POVM}\eqdef \sum_{x}\frac{\vvec{M_x}\vadj{M_x}}{\Tr[M_x]} \in \C^{\dims^2\times\dims^2},
\end{equation}
where $\vvec{M_x}=\VecOp(M_x)$ and $\vadj{M_x}=\VecOp(M_x)^{\dagger}$.
\end{definition}
The channel maps a quantum state to another quantum state. Intuitively, it characterizes the similarity of the outcome distributions after applying $\POVM$ to $\rho$ and the maximally mixed state $\qmm$. 
The ability to test against the maximally mixed state is described by the eigenvalues of the channel. 
In~\cref{thm:lower-channel}, we show a unified lower bound for non-adaptive measurements. 
Depending on the availability of randomness, the bound depends on different norms of $\Luders_{\POVM}$.  We summarize~\cref{thm:lower-channel} and the existing and new results for non-adaptive state certification in~\cref{tab:results}. 
\begin{theorem}
\label{thm:lower-channel}
    Let $\povmset$ be a set of allowable POVMs on a $\dims$-dimensional system. Let $\ns_R$ and $\ns_F$ be the copy complexity of mixedness testing for randomized and fixed unentangled measurements from $\povmset$ respectively. Then for $\dims\ge 16$ and $\eps\le 1/200$,
    \[
    \ns_R=\bigOmega{\frac{\dims^2}{\eps^2\sup_{\POVM\in\povmset }\hsnorm{\Luders_{\POVM}}}},\quad \ns_F=\bigOmega{\frac{\dims^2}{\eps^2}\cdot \frac{\dims}{\sup_{\POVM\in\povmset }\tracenorm{\Luders_{\POVM}}}},
    \]
    where $\hsnorm{\Luders_{\POVM}}$ is the Hilbert-Schmidt/Frobenius norm and  $\tracenorm{\Luders_{\POVM}}$ is  the trace norm, which are $\ell_2$ and $\ell_1$ norms of the spectrum respectively.
\end{theorem}

\begin{table}
    \def\arraystretch{1.6}
        \centering
        \begin{tabular}{|c | c | c | c | c |}
        \hline
       Scheme  & Restricted (lower) & Pauli & $2\le  \ab<\dims$ & $\ab\ge \dims$ \\ \hline
             {Randomized} & $\frac{\dims^2}{\eps^2\max_{\POVM}\hsnorm{\Luders_{\POVM}}}$ & {$\frac{\dims^2}{\eps^2}^\dagger$} & {\color{blue}$\frac{\dims^2}{\eps^2\sqrt{k}}$} & {$\frac{\dims^{3/2}}{\eps^2}$}   \\ \hline
             {Fixed} &$\frac{\dims^2}{\eps^2}\cdot \frac{\dims}{\max_{\POVM}\tracenorm{\Luders_{\POVM}}}$ & {\color{blue}$\frac{\dims^3}{\eps^2}$} &{\color{blue}$\frac{\dims^3}{\eps^2{\ab}}$}    &{\color{blue}$\frac{\dims^2}{\eps^2}$}  \\\hline
        \end{tabular}
        \caption{Existing and new worst-case copy complexity for quantum state certification with non-adaptive measurements. All results are constant-optimal unless otherwise indicated. New results are highlighted in blue. $\dagger$: tight up to log factors.}
        \label{tab:results}
\end{table}

By Cauchy-Schwarz, $\hsnorm{\Luders_{\POVM}}\ge \tracenorm{\Luders_{\POVM}}/\sqrt{\dims^2}$ (because the spectrum has dimension $\dims^2$). 
Thus the lower bound for $\ns_R$ is smaller than that of $\ns_F$, 
which is consistent with that randomized measurements are more powerful than fixed ones. 
This plug-and-play result allows us to easily obtain lower bounds for any measurement set $\povmset$, as long as we can upper bound the respective norms of the MICs in $\povmset$. For $\ab$-outcome measurements, the bounds are stated in~\cref{lem:finite-outcome-norms}.
\begin{lemma}
\label{lem:finite-outcome-norms}
    For a POVM $\POVM$ with at most $\ab$ outcomes, $\hsnorm{\Luders_{\POVM}}^2\le \tracenorm{\Luders_{\POVM}}\le \min\{\ab,\dims\}.$
\end{lemma}
Combining~\cref{lem:finite-outcome-norms} and~\cref{thm:lower-channel} proves all lower bounds in~\cref{thm:random-k,thm:fixed-k-lower}. 

\paragraph{Relation to \cite{liu2024role}.} When $M_x$ is rank-1, i.e. $M_x=\qproj{\psi_x}$, the MIC 
 $\Luders_{\POVM}$ is exactly the L\"uders channel~\citep{debrota2019luders} which describes one possible form of expected post-measurement state after measuring with $\POVM$. However, their arguments using L\"uders channel only hold for rank-1 POVMs. Our formulation using MIC is a necessary extension of their argument to arbitrary POVMs which may not be rank-1.

\subsection{The disadvantages of fixed measurements}
\label{sec:fixed-disadvantage}
One of the key findings in \cref{sec:new-results} is the separation between fixed and randomized measurements. We explain the high-level reason why such separation exists through a simple example.

In randomized schemes, given the copies of the state, we then choose the measurements randomly. However, for fixed measurements, the measurement scheme is fixed and the state is then chosen by an imaginary adversary. Thus, without randomness, nature would have the opportunity to \textit{adversarially} design a quantum state that fools the pre-defined set of measurements. When shared randomness is available, we can avoid the bad effect of adversarial choice of quantum states. In principle, this qualitative gap is like the difference between randomized algorithms and deterministic algorithms.

We use a simple example to demonstrate this idea. Suppose we choose each measurement $\POVM_i$ simply to be the same canonical basis measurement, i.e. $\POVM_i=\{\qproj{x}\}_{x=0}^{\dims-1}$. Then nature can set $\rho$ to be the ``+'' state where
\begin{equation}
    \rho = \qproj{\phi},\quad \qbit{\phi} = \frac{1}{\sqrt{\dims}}\sum_{x=0}^{\dims-1}\qbit{x}.
    \label{equ:plus-state}
\end{equation}

Note that the trace distance $\frac{1}{2}\tracenorm{\rho-\qmm}=1-1/\dims\simeq 1$. When the underlying state is $\rho$, all measurement outcomes $x_i$ would be independent samples from the uniform distribution over $\{0, \ldots, d-1\}$. However, if the state is the maximally mixed state $\qmm$, the distribution of each measurement outcome would also be the uniform distribution over $\{0, \ldots, d-1\}$. Thus, even though the trace distance between $\rho$ and $\qmm$ is large, the measurement scheme is completely fooled.

On the other hand, with shared randomness, one can (theoretically) sample a basis uniformly from the Haar measure as in \cite{BubeckC020} to easily avoid this issue. No fixed $\rho$ would be able to completely fool the randomized basis measurement sampled uniformly. In fact with high probability, the randomly sampled basis is good in the sense that the outcome distribution would be far enough when the two states $\rho$ and $\qmm$ are far (see~\cite[Lemma 6.3]{ChenLO22instance}).

Similar to \cite{AcharyaCT19}, we develop a min-max and max-min framework to formally characterize the difference between fixed and randomized measurements.

\subsection{A novel lower bound construction} 
\label{sec:summary:new-construction}
The design of hard instances has to account for the difference illustrated above when proving lower bounds for randomized and fixed measurements. In particular, for randomized measurements, the lower bound construction can be \emph{measurement independent}. However, for fixed measurements, since the states can be chosen adversarially, the lower bound construction needs to be \emph{measurement-dependent}.

Many prior works on testing and tomography~\cite{ODonnellW15,ODonnellW16,HaahHJWY17, BubeckC020, ChenCH021,Chen0HL22} use \textit{measurement-independent} distributions over states in $\mathbb{C}^{\dims\times\dims}$ to prove lower bounds. 
In particular, \cite{BubeckC020, Chen0HL22} show that testing within a specific class requires at least $\ns=\Omega(d^{3/2}/\dst^2)$ when working with \emph{randomized} and \emph{adaptive} unentangled measurements respectively. 
Unfortunately, these measurement-independent constructions are unable to capture the disadvantage of fixed measurements illustrated in our simple example in \cref{sec:fixed-disadvantage}. We note that the lower bound construction in~\cite{Yu2023almost} is measurement-dependent, but specifically tailored to Pauli measurements and not general enough for our purpose.

Our generic \textit{measurement-dependent} lower bound construction is a necessary and novel contribution that leads to tight lower bounds for state certification with fixed measurements. It takes the form
\begin{equation}
\label{equ:quantum-construction-informal}
    \sigma_z=\qmm + \frac{\dst}{\sqrt{\dims}}\cdot\frac{\cd}{\dims}\sum_{i=1}^{\dims^2/2} z_iV_i,
\end{equation}
where $\{V_i\}_{i=1}^{\dims^2-1}$ are $\dims^2-1$ orthonormal trace-0 Hermitian matrices, $z=(z_1, \ldots, z_{\dims^2/2})$ are uniformly sampled from $\{-1, 1\}^{\dims^2/2}$, and $\cd$ is an absolute constant. 
In essence, we perform independent binary perturbations along different trace-0 directions. We show with appropriate choice of $c$, regardless of the choice of $\{V_i\}_{i=1}^{\dims^2-1}$, 
with high probability over the randomness of $z$, $\sigma_z$ is a valid quantum state and $\dst$-far in trace distance from $\qmm$.
The matrices $V_i$'s can be chosen \emph{dependent} on the fixed measurement scheme that we want to \emph{fool}.
In particular, the perturbations $V_1, \ldots, V_{\dims^2/2}$ can be chosen in directions about which the fixed measurement schemes provide the least information.
The matrices $V_i$'s can also be fixed, in which case the construction is measurement-independent and our framework naturally leads to the lower bound for randomized non-adaptive measurements in~\cite{BubeckC020}.

\paragraph{Relation to Paninski's construction~\cite{Paninski08}.} Our construction can be viewed as a generalization of Paninski's construction for classical discrete distribution testing where the hard instances are constructed as perturbations around the uniform distribution. Assume $\dims$ is even and let $z\in\{-1, 1\}^{\dims/2}$. For a distribution over $[\dims]$ and some constant $c$,
\begin{align}
\p_z= \frac{1}{\dims} (1+c\dst z_1, 1-c\dst z_1, &\ldots, 1+c\dst
z_t, 1-4\dst z_t, \notag\\
&\ldots, 1+c\dst z_{\dims/2}, 1-c\dst z_{\dims/2} )\,. \label{eq:paninski}
\end{align}
We can see that each $z$ defines a perturbation $\mathbf{d}_z=\p_z-\unif$ in the space of \emph{sum-zero} vectors, and the number of such perturbations is $\dims/2$, the same order as the dimensionality of discrete distributions over $[\dims]$. In \eqref{equ:quantum-construction-informal}, we make \emph{trace-zero} perturbations to the maximally mixed state along $\Theta(\dims^2)$ different directions. Recall that the dimensionality of quantum states is $\dims^2$, and maximimally mixed state generalizes uniform distributions, our construction \eqref{equ:quantum-construction-informal} can be viewed as a generalization of \eqref{eq:paninski}.

\paragraph{Relation to previous quantum constructions.} The binary perturbations in our construction are mathematically easier to handle than previous works. \cite{BubeckC020, HaahHJWY17, ChenCH021} designed the hard cases using random unitary transformations around the maximally-mixed state, which requires difficult calculations using Weingarten calculus~\cite{weingarten1978asymptotic,collins2003moments}. In contrast, our arguments avoid the difficult representation-theoretic tools. \cite{Chen0HL22, chen2023does} used Gaussian orthogonal ensembles, which perturbs each matrix entry with independent Gaussian distributions. Binary perturbations share many statistical similarities with Gaussian since both are sub-gaussian distributions. However, the former is arguably simpler as the support is finite, and thus information-theoretic tools can be more easily applied. We note that all these constructions are somewhat inspired by Paninski's construction in \cite{Paninski08}.

\section{Preliminaries}
\label{sec:preliminaries}

\subsection{Classical distribution testing}

\paragraph{Probability distances and divergences} Let $\p$ and $\q$ be distributions over a finite domain $\mathcal{X}$. The \emph{total variation distance} is defined as 
$$
\totalvardist{\p}{\q}\eqdef\sup_{S\subseteq\mathcal{X}}(\p(S)-\q(S))=\frac{1}{2}\sum_{x\in\mathcal{X}}|\p(x)-\q(x)|.$$

The KL-divergence is
$$\kldiv{\p}{\q}\eqdef\sum_{x\in\mathcal{X}}\p(x)\log\frac{\p(x)}{\q(x)}.$$

The chi-square divergence is
$$\chisquare{\p}{\q}\eqdef \sum_{x\in\mathcal{X}}\frac{(\p(x)-\q(x))^2}{\q(x)}.$$

By Pinsker's inequality and concavity of logarithm,
\[
2\totalvardist{\p}{\q}^2\le \kldiv{\p}{\q}\le \chisquare{\p}{\q}.
\]
We may also define $\ell_p$ distances between distributions, $
\norm{\p-\q}_p\eqdef\Paren{\sum_{x\in\mathcal{X}}{|\p(x)-\q(x)|^p}}^{1/p}.
$

\paragraph{Distribution testing under $\ell_2$ distance} Testing discrete distributions under $\ell_2$ distance is a well studied problem. The algorithm and its sample complexity is stated as follows.
\begin{theorem}[{\citet[Lemma 2.3]{DiakonikolasK16}}]
     \label{thm:distr-testing}
     Let $\p$ be an unknown distribution and $\q$ be a known target distribution over $[k]$ such that $\min\{\normtwo{p}, \normtwo{q}\}\le b$. Let $\bx=(x_1, \ldots, x_\ns)$ be $\ns$ i.i.d. samples from $\p$. There exists an algorithm TestIdentityL2($\q, \bx,\eps, \delta$) that outputs $\isthestate$ if $\p=\q$ and $\notthestate$ if $\normtwo{\p-\q}>\eps$ with probability at least $1-\delta$ using $\ns=1000b\log(1/\delta)/\eps^2$ samples. 
\end{theorem}

\subsection{Hibert space over linear operators}

\paragraph{Hilbert space over complex matrices.}
The space of complex matrices $\C^{\dims\times\dims}$ is a Hilbert space with inner product 
$$\hdotprod{A}{B}\eqdef\Tr[A^\dagger B], A, B\in \C^{\dims\times \dims}.$$
For Hermitian matrices $A,B$, $\hdotprod{A}{B}=\hdotprod{B}{A}\in \R$. Thus the subspace of Hermitian matrices $\Herm{\dims}$ is a \textit{real} Hilbert space (i.e. the associated field is $\R$) with the same matrix inner product. 

Vectorization defines a homomorphism between $\C^{\dims\times\dims}$ and $\C^{\dims^2}$, where 
$$\VecOp(\qoutprod{i}{j})\eqdef \qbit{j}\otimes \qbit{i}.$$
Vectorization for a general matrix $A$ is defined through linearity. For convenience we denote $\vvec{A}\eqdef\VecOp(A)$. Matrix inner product can be written as inner product on $\C^{\dims^2}$, $\hdotprod{A}{B}=\vvdotprod{A}{B}$. 

\paragraph{(Linear) superoperators.} Let $\mathcal{N}:\C^{\dims\times \dims}\mapsto \C^{\dims\times \dims}$ be a linear operator over $\C^{\dims\times \dims}$, which we refer to as superoperators\footnote{This is to distinguish from a matrix $A\in\C^{\dims\times \dims}$, which can be viewed as an operator over $\C^\dims$. Indeed an operator over $\C^{\dims\times \dims}$ need not be linear, but we only deal with linear ones in this work, so we drop ``linear'' for brevity.}. Every superoperator $\mathcal{N}$ has a matrix representation $\Choi(\mathcal{N})\in \C^{\dims^2\times\dims^2}$ that satisfies 
$$\vvec{\mathcal{N}(X)}=\Choi(\mathcal{N})\vvec{X}$$ 
for all matrices $X\in\C^{\dims\times\dims}$. It can be verified that for the measurement information channel $\Luders_{\POVM}$ in~\cref{def:mic}, $\Choi_{\POVM}\vvec{A}=\vvec{\Luders_{\POVM}(A)}$.

\paragraph{Schatten norms.} Let $\Lambda=(\lambda_1, \ldots, \lambda_\dims)\ge 0$ be the \emph{singular values} of a linear operator $A$, which can be a matrix or a superoperator. {For Hermitian matrices, the singular values are the absolute values of the eigenvalues.} Then for $p\ge 1$, the \emph{Schatten $p$-norm} is defined as 
$$
\|A\|_{S_p}\eqdef \|\Lambda\|_p.
$$ 
The Schatten norms of a superoperator $\mathcal{N}$ and its matrix representation $\Choi(\mathcal{N})$ match exactly, $\|\mathcal{N}\|_{S_p}=\|\Choi(\mathcal{N})\|_{S_p}$. Some important special cases are trace norm $$\tracenorm{A}\eqdef\|A\|_{S_1},$$
Hilbert-Schmidt norm 
$$\hsnorm{A}\eqdef\|A\|_{S_2}=\sqrt{\hdotprod{A}{A}},$$
and operator norm 
$$\opnorm{A}\eqdef\|A\|_{S_\infty}=\max_{i=1}^\dims\lambda_i.$$
Due to Cauchy-Schwarz and monotonicity of $\ell_p$ norms,
\[
\hsnorm{A}\le \tracenorm{A}\le \sqrt{\dims}\hsnorm{A}.
\]
By H\"older's inequality,
\[
\hsnorm{A}^2\le \opnorm{A}\tracenorm{A}.
\]

\section{The measurement information channel (MIC)}
\label{sec:MIC}
Recall for a measurement $\POVM$, the measurement information channel is

$$\Luders_{\POVM}(\rho)\eqdef\sum_{x}M_x\frac{\Tr[M_x\rho]}{\Tr[M_x]}.$$

It is a \emph{measure and prepare channel}~\cite[Eq (4.4.82)]{khatri2021principles} where upon measuring a quantum state $\rho$ with $\POVM$ and observing an outcome $x$, the system prepares the quantum state $M_x/\Tr[M_x]$. Its matrix representation is 
\[
\Choi_{\POVM}\eqdef \sum_{x}\frac{\vvec{M_x}\vadj{M_x}}{\Tr[M_x]} \in \C^{\dims^2\times\dims^2},
\]
which satisfies $\Choi_{\POVM}\vvec{A}=\vvec{\Luders_{\POVM}(A)}$ for all $\dims\times\dims$ matrix $A$.

In this section, we first discuss some important properties of MIC. Then we study a special case that links MIC to a natural physical quantity. Finally, we use a simple example to show that the MIC of a measurement characterizes the power of distinguishability of this measurement.

\subsection{Properties of MIC}
We introduce important properties of the MIC in~\cref{fact:mic-properties}. The proof is in~\cref{app:fact:mic-properties}.
\begin{lemma}
\label{fact:mic-properties}
    Let $\POVM$ be a POVM and $\Luders_{\POVM}$ be the MIC with matrix representation $\Choi_{\POVM}$. Then
    \begin{enumerate}
        \item $\Choi_{\POVM}$ is positive semi-definite.
        \item $\Luders_{\POVM}$ is unital, i.e. $\Luders_{\POVM}(\eye_\dims)=\eye_\dims$.
        \item $\Luders_{\POVM}$ is trace-preserving, i.e. $\Tr[\Luders_{\POVM}(X)]=\Tr[X]$.
        \item $\Luders_{\POVM}$ is Hermitian preserving, i.e. for all Hermitian $X$, $\Luders_{\POVM}(X)$ is also Hermitian.
    \end{enumerate}
\end{lemma}
 An immediate corollary is that $\Luders_{\POVM}$ has an eigen-decomposition where one eigenvector is $\eye_\dims$ and all other eigenvectors are trace-0 Hermitian matrices.
\begin{corollary}
\label{cor:mic-eigenbasis}
    For all POVM $\POVM$, there exists Hermitian matrices $\hbasis=(V_1, \ldots, V_{\dims^2})$ with $\Tr[V_j]=0$ for all $j\le \dims^2-1$ and $V_{\dims^2}=\eye_\dims/\sqrt{\dims}$ which form an orthonormal eigenbasis of $\Luders_{\POVM}$. 
\end{corollary}
\begin{proof}
    Since $\Luders_{\POVM}$ is Hermitian preserving, it is a linear map from $\Herm{\dims}$ to $\Herm{\dims}$. The dimension of $\Herm{\dims}$ is $\dims^2$. Thus, there exists $\dims^2$ matrices $V_1, \ldots, V_{\dims^2}\in \Herm{\dims}$ which are eigenvectors of $\Luders_{\POVM}$.

    Since $\Luders_{\POVM}$ is unital, $\eye_\dims$ is an eigenvector of $\Luders_{\POVM}$ with eigenvalue of 1. The eigenspace orthogonal to $\eye_\dims$ is exactly the trace-0 subspace because $\hdotprod{\eye_\dims}{X}=0\iff \Tr[\eye_\dims X]=\Tr[X]=0$. Setting $V_{\dims^2}=\eye_\dims/\sqrt{\dims}$ and with appropriate normalization for the trace-0 eigenvectors, we obtain an orthonormal basis.
\end{proof}

We upper bound the operator norm in~\cref{lem:mic-opnorm} and the trace norm in~\cref{lem:mic-trace}. The proofs are in~\cref{app:lem:mic-opnorm} and~\cref{app:lem:mic-trace} respectively.

\begin{lemma}
\label{lem:mic-opnorm}
    For all POVM $\POVM$, $\opnorm{\Luders_{\POVM}}\le 1$.
\end{lemma}

\begin{lemma}
\label{lem:mic-trace}
    For all POVM $\POVM$ with at most $\ab$ outcomes, $\tracenorm{\Luders_{\POVM}}\le \min\{\dims, \ab\}$. 
\end{lemma}

 By H\"older's inequality, $\hsnorm{\Luders_{\POVM}}^2\le \opnorm{\Luders_{\POVM}}\tracenorm{\Luders_{\POVM}}\le \min\{\dims, \ab\}$, which proves \cref{lem:finite-outcome-norms}.

\subsection{Special case: L\"uders channel}
Consider the family of rank-1 measurements $\POVM=\{M_x\}$ where each $M_x=\qproj{\psi_x}$ is rank-1. Note that $\qbit{\psi_x}$ may not be normalized. In the \emph{unrestricted case} where we can choose arbitrary measurements for each copy, it is without loss of generality to only consider rank-1 measurements\footnote{This is because if some $M_x$ is not rank-1, then we can apply eigen-decomposition and write it as a sum of rank-1 projection matrices, $M_x=\sum_{j}\qproj{\phi_{j,x}}$. Replacing $M_x$ with $\{\qproj{\phi_{j,x}}\}_j$ only increases the power of the measurement because one outcome $x$ can be represented using the collection of outcomes $\{(j, x)\}_j$. See \cite[Lemma 4.8]{ChenCH021} for a formal proof.}. In this case, the MIC evaluates as
    \begin{equation}
        \Luders_{\POVM}(\rho) = \sum_{x=1}^{k}\frac{\qproj{\psi_x}\rho\qproj{\psi_x}}{\qdotprod{\psi_x}{\psi_x}} =  \sum_{x=1}^kK_x\rho K_x,\; K_x=\frac{\qproj{\psi_x}}{\sqrt{\qdotprod{\psi_x}{\psi_x}}}.
        \label{equ:rank-1-luders-channel}
    \end{equation}
Note that $K_x^2=M_x$ so that $K_x$ is the unique p.s.d. square-root of $M_x$. 

In fact, \eqref{equ:rank-1-luders-channel} has a natural physical interpretation. When $\POVM$ acts on $\rho$ and we obtain an outcome $x$, then the \textit{generalized L\"uder's rule}~\citep{luders1950zustandsanderung} gives one possible form of the post-measurement state\footnote{The post-measurement states are undefined for general POVMs.},
\[
\rho^x\eqdef\frac{K_x\rho K_x}{\Tr[M_x\rho]}, \quad K_x=\sqrt{M_x}.
\]
$\Luders_{\POVM}(\rho) $ is exactly the expectation of $\rho^x$ when $x$ follows the outcome distribution defined by the Born's rule,
\[
\Luders_{\POVM}(\rho)=\sum_{x}\rho^x\Tr[M_x\rho].
\]
Thus, when $\POVM$ is rank-1, $\Luders_{\POVM}$ characterizes a natural physical process which describes one possible post-measurement state transition when the measurement outcome is not available, which is called L\"uders channel \cite{debrota2019luders,luders1950zustandsanderung}. As \cite{liu2024role} proved, this channel characterizes the fundamental limit of testing: If two input states $\rho,\rho'$ lead to the same post-measurement state, then they are physically indistinguishable after measurement, let alone any algorithms that only rely on the measurement outcomes.

\subsection{MIC characterizes distinguishability: a toy example}
\label{sec:toy-example}
We use the same example in \cref{sec:fixed-disadvantage} to demonstrate how MIC characterizes the distinguishability of any fixed measurement. Recall that if we choose $\POVM=\{\qproj{x}\}_{x=0}^{\dims-1}$, then nature can set $\rho$ to be the ``+'' state where
\begin{equation*}
    \rho = \qproj{\phi},\quad \qbit{\phi} = \frac{1}{\sqrt{\dims}}\sum_{x=0}^{\dims-1}\qbit{x}.
\end{equation*}

We argued that both $\rho$ and $\qmm$ yields the uniform distribution when measured with $\POVM$ but their trace distance is nearly as large as it can be. Now we look at how MIC plays a role in this argument. For $\POVM=\{\qproj{x}\}_{x=0}^{\dims-1}$, the MIC is 
$$\Luders_{\POVM}(\cdot)=\sum_{x=0}^{\dims-1}\qproj{x}(\cdot)\qproj{x}.$$

It turns out that $\rho-\qmm$ exactly falls into the 0-eigenspace of $\Luders$,
\begin{align*}
    \Luders_{\POVM}(\rho-\qmm)&=\sum_{x=0}^{d-1}\qproj{x}(\rho-\qmm)\qproj{x}\\
    &=\sum_{x=0}^{d-1}\qbit{x}\qdotprod{x}{\phi}\qdotprod{\phi}{x}\qadjoint{x}-\sum_{x=0}^{d-1}\qproj{x}\frac{\eye_d}{d}\qproj{x}\\
    &=\sum_{x=0}^{d-1}\qbit{x}\qadjoint{x}\frac{1}{d}-\frac{\eye_\dims}{\dims}=0.
\end{align*}

It is reasonable to believe that for any measurement $\POVM$, if the difference of two states $\rho-\rho'$ fall into an eigenspace of $\Luders_{\POVM}$ with small eigenvalues, then it is hard to distinguish $\rho,\rho'$ with $\POVM$.

This toy example serves as a good intuition about the role of MIC in testing lower bound, but lacks rigorousness and does not generalize to randomized measurements and the case where measurement applied to each copy can be different. In the following sections, we will formalize this intuition by relating the hardness of testing to the MIC for general non-adaptive measurements.

\section{Lower bound framework using MIC}
\label{sec:lower}
This section is dedicated to proving the key lower bound result in~\cref{thm:lower-channel}. We focus on mixedness testing where $\qkn=\qmm\eqdef\eye_\dims/\dims$. For an arbitrary state $\rho$ we define $\Delta_\rho\eqdef\rho-\qmm$.

\subsection{The min-max and max-min chi-square divergences}
\label{sec:min-max-max-min}

We use the Le Cam's method~\cite{LeCam73,yu1997assouad} to prove lower bounds. Recall that given measurement scheme $\POVM^n$ and a state $\rho$,  the distribution of the outcomes $\bx$ is $\bfP_{\rho}$. Let $\mathcal{P}_\eps\eqdef\{\rho:\tracenorm{\qs-\qmm}>\eps\}$ be states at least $\eps$ far from $\qmm$ in trace distance. We define \emph{almost-$\eps$ perturbations}, which have constant probability mass over $\mathcal{P}_\eps$.

\begin{definition}
    A distribution $\mathcal{D}$ over quantum states is an \emph{almost-$\eps$ perturbation} if $\probaDistrOf{\sigma\sim\mathcal{D}}{\sigma\in\mathcal{P}_\eps}>\frac{1}{2}.$ Denote the set of \emph{almost-$\eps$ perturbation} as $\Gamma_\eps$.
\end{definition}

Consider a two-party game played between nature and an algorithm designer. Nature flips a coin $Y$ where $Y=0$ or $1$ with probability 1/2. If $Y=1$, then nature sets $\rho=\qmm$. Applying $\POVM^n$ on $\ns$ copies, the outcomes $\bx\sim \bfP_{\qmm}$. If $Y=0$, nature chooses $\mathcal{D}\in \Gamma_\eps$ of its choice and samples $\rho\sim\mathcal{D}$. In this case, the outcomes $\bx\sim \expectDistrOf{\rho\sim\mathcal{D}}{\bfP_{\rho}}$. Using a mixedness tester, we should guess $Y$ correctly with constant probability. By Le Cam's and Pinsker's inequality, this means that $\chisquare{\expectDistrOf{\rho\sim\mathcal{D}}{\bfP_{\rho}}}{\bfP_{\qmm}}$ cannot be too small. \cref{lem:min-max-and-max-min} states the precise necessary conditions for the existence of a mixedness tester using fixed and randomized measurements.

\begin{lemma}[{\cite[Lemma IV.8, IV.10]{AcharyaCT19}}]
\label{lem:min-max-and-max-min}
If there exists a tester using fixed non-adaptive measurements with a success probability at least $2/3$, then
   \begin{equation}
   \label{equ:max-min-lb}
        \frac{2}{25}\le \max_{\mathcal{M}^n\text{ fixed}}\min_{\mathcal{D}\in \Gamma_\eps}\chisquare{\expectDistrOf{\sigma\sim\mathcal{D}}{\bfP_{\sigma}}}{\bfP_{\qmm}}.
   \end{equation}
If there exists a tester using randomized non-adaptive measurements with 2/3 success rate, then
   \begin{equation}
\label{equ:min-max-lb}
    \frac{2}{25}\le \min_{\mathcal{D}\in \Gamma_\eps}\max_{\mathcal{M}^n\text{ fixed}}\chisquare{\expectDistrOf{\sigma\sim\mathcal{D}}{\bfP_{\sigma}}}{\bfP_{\qmm}}.
\end{equation}
\end{lemma}
\begin{proof}
      Recall that $Y=0$ and $\rho=\qmm$ with probability 1/2 and $Y=1$ and $\rho\sim \mathcal{D}$ with probability $1/2$. First, we prove that for any measurement scheme $\POVM^\ns$, the chi-square divergence of outcome distributions should be far for $Y=0$ and $Y=1$. 
    
    In the former case when the state is $\qmm$ and $Y=0$, then the tester outputs the correct answer with probability at least $2/3$, 
    \[
    \Pr[\hat{Y}=0|Y=0]\ge 2/3.
    \]
    
    When $\rho\sim\mathcal{D}$, note that by the definition of almost-$\eps$ perturbations, the probability that $\|\sigma_z-\qmm\|_1>\eps$ is at least $4/5$. Denote this event as $E$, then $\Pr[E|Y=1]\ge 4/5$ . We can lower bound the success probability as
    \[
    \Pr[\hat{Y}=1|Y=1]\ge \Pr[Y=1|E, Y=1)]\Pr[E|Y=1]\ge \frac{2}{3}\cdot \frac{4}{5}=\frac{8}{15}.
    \]
    Combining the two parts,
    \[
    \Pr[Y=\hat{Y}]=\frac{1}{2}\Pr[\hat{Y}=0|Y=0]+\frac{1}{2}\Pr[\hat{Y}=1|Y=1]\ge \frac{1}{2}\Paren{\frac{2}{3}+\frac{8}{15}}=\frac{3}{5}.
    \]
By Le Cam's method~\cite{LeCam73,yu1997assouad}, 
    $$1-\frac{3}{5}\ge \frac{1}{2}(1-\totalvardist{\expectDistrOf{z}{\bfP_{\sigma_z}}}{\bfP_{\qmm}}) \implies\totalvardist{\expectDistrOf{z}{\bfP_{\sigma_z}}}{\bfP_{\qmm}}\ge \frac{1}{5}. $$
    
    Using Pinsker's inequality and the relation between KL and chi-square divergences, we have,
    \begin{equation*}
        \frac{1}{5}\le \totalvardist{\expectDistrOf{\sigma}{\bfP_{\sigma}}}{\bfP_{\qmm}}\le \sqrt{\frac{1}{2}\kldiv{\expectDistrOf{\sigma\sim\mathcal{D}}{\bfP_{\sigma}}}{\bfP_{\qmm}}}\le\sqrt{\frac{1}{2}\chisquare{\expectDistrOf{\sigma\sim\mathcal{D}}{\bfP_{\sigma}}}{\bfP_{\qmm}}}.
    \end{equation*}
    Rearranging the terms,
    \begin{equation}
    \label{equ:chi-square-lb}
        \frac{2}{25}\le \chisquare{\expectDistrOf{\sigma\sim\mathcal{D}}{\bfP_{\sigma}}}{\bfP_{\qmm}}.        
    \end{equation}

Now we apply \eqref{equ:chi-square-lb} to fixed and randomized measurements to arrive at the min-max and max-min expressions.

For a fixed measurement scheme $\POVM^\ns$, nature can choose a $\mathcal{D}\in \Gamma_\eps$ that minimizes the chi-square divergence in~\eqref{equ:chi-square-lb}.
According to~\eqref{equ:chi-square-lb}, if there exists a fixed $\POVM^{\ns}$ that achieves at least 2/3 probability in testing maximally mixed states, we must have
\begin{equation*}
    \frac{2}{25}\le \max_{\mathcal{M}^n\text{ fixed}}\min_{\mathcal{D}\in \Gamma_\eps}\chisquare{\expectDistrOf{\sigma\sim\mathcal{D}}{\bfP_{\sigma}}}{\bfP_{\qmm}}.
\end{equation*}
Thus a max-min game is played between the two parties and nature has an advantage to decide its best action based on the choice of the algorithm designer.

With randomness, in principle, a max-min game is still played, but instead, the maximization is over all \textit{distributions} of fixed (non-entangled) measurements.  Using a similar argument as~\cite[Lemma IV.8]{AcharyaCT19}, for the best distribution over all $\mathcal{M}^n$,  the expected accuracy over $R\sim\mathcal{R}$ is at least 1/2 for all $\mathcal{D}\in \Gamma_\eps$. Thus, for all $\mathcal{D}$, there must exist an instantiation $R(\mathcal{D})$ such that using the fixed measurement $\mathcal{M}^n(R(\mathcal{D}))$ the testing accuracy is at least 1/2. Therefore, 
\begin{equation*}
    \frac{2}{25}\le \min_{\mathcal{D}\in \Gamma_\eps}\max_{\mathcal{M}^n\text{ fixed}}\chisquare{\expectDistrOf{\sigma\sim\mathcal{D}}{\bfP_{\sigma}}}{\bfP_{\qmm}},
\end{equation*}
which intuitively says that a min-max game is played and the algorithm designer has an advantage. 
\end{proof}

The intuition of \cref{lem:min-max-and-max-min} is that nature and the algorithm designer play a two-party game, 
where the algorithm designer tries to maximize the chi-square divergence, while nature picks $\mathcal{D}\in \Gamma_\eps$ that minimizes it and fools the algorithm. For fixed measurements, nature can choose the distribution $\mathcal{D}\in \Gamma_\eps$ adversarially depending on the measurement scheme. With randomness, measurements are instantiated after the states are given, which puts the algorithm designer at a second-mover advantage.
The min-max and max-min arguments are similar to \cite{AcharyaCT19} and we point to \cite[Lemma IV.8, IV.10]{AcharyaCT19} for additional reference.

Therefore, to obtain a copy complexity lower bound for fixed measurements requires upper bounding~\eqref{equ:max-min-lb}, while for randomized schemes requires upper bounding~\eqref{equ:min-max-lb}. We can see that randomness is a ``game changer'' that changes a max-min game to a min-max game. Since min-max is no smaller than max-min, testing with randomness is easier than testing without it. This formalizes the intuition we established from the toy example in \cref{sec:fixed-disadvantage}.

\subsection{Relating chi-squared divergence to the MIC}
\label{sec:avg-luders}
Our central contribution is that we relate the min-max and max-min divergences to the \textit{average measurement information channel} defined by $\POVM^\ns$. Define the shorthand $\Luders_i\eqdef\Luders_{\POVM_i}$ and $\Choi_i\eqdef\Choi_{\POVM_i}$, the channel $\avgLuders$ and its matrix representation $\avgChoi$ are defined as
\begin{equation}
\avgLuders=\frac{1}{\ns}\sum_{i=1}^{\ns}\Luders_i,\quad \avgChoi\eqdef\frac{1}{\ns}\sum_{i=1}^{\ns}\Choi_i.
    \label{equ:average-luders-channel}
\end{equation}
Due to linearity, \cref{fact:mic-properties} and \cref{cor:mic-eigenbasis} hold for $\avgLuders,\avgChoi$. 
The key technical result is \cref{lem:decoupled-chi-square} which upper bounds $\chisquare{\expectDistrOf{\sigma\sim\mathcal{D}}{\bfP_{\sigma}}}{\bfP_{\qmm}}$ for fixed measurements using the average MIC. 
\begin{lemma}
\label{lem:decoupled-chi-square}
Let $\sigma, \sigma'$ be independently drawn from a distribution $\mathcal{D}$, and $\POVM^n=(\POVM_1, \ldots, \POVM_\ns)$ be a fixed measurement scheme. Define $\Delta_\sigma = \sigma -\qmm$. Then
    \begin{align}
    \label{equ:chi-square-quantum}
        \chisquare{\expectDistrOf{\sigma\sim\mathcal{D}}{\bfP_{\sigma}}}{\bfP_{\qmm}}&\le \expectDistrOf{\sigma, \sigma'\sim \mathcal{D}}{\exp\left\{\ns\dims \cdot\vadj{\Delta_\sigma}\avgChoi\vvec{\Delta_{\sigma'}} \right\}}-1
    \end{align}
    where  $\avgChoi$ is the matrix representation of the average measurement information channel.
\end{lemma}
\begin{proof}
    We can directly bound the chi-square divergence using the following lemma.
    \begin{lemma}[\cite{Pollard:2003}]
    \label{lem:chi-square-expansion}
        Let $\bfP=\p\supparen{1}\otimes \cdots\otimes \p\supparen{n}$ be a fixed product distribution and $\bfQ_\theta=\q_{\theta}\supparen{1}\otimes \cdots \otimes \q_{\theta}\supparen{n}$ be parameterized by a random variable $\theta$. Then
        \[
        \chisquare{\expectDistrOf{\theta}{\bfQ_{\theta}}}{\bfP}=\expectDistrOf{\theta, \theta'}{\prod_{i=1}^n(1+H_i(\theta, \theta'))}-1,
        \]
        where $\theta'$ is an independent copy, and 
        $
        H_i(\theta, \theta')\eqdef \expectDistrOf{x\sim \p\supparen{i}}{\delta_\theta\supparen{i}(x)\delta_{\theta'}\supparen{i}(x)},\;\delta_{\theta}\supparen{i}(x)\eqdef\frac{\q_\theta\supparen{i}(x)-\p\supparen{i}(x)}{\p\supparen{i}(x)}.
        $
    \end{lemma}
    We apply \cref{lem:chi-square-expansion} with $\bfP=\bfP_{\qmm}$ and $\bfQ_\sigma=\bfP_{\sigma}$, so that $\delta_{\sigma}\supparen{i}(x)=\Paren{\p_{\sigma}\supparen{i}(x)-\p_{\qmm}\supparen{i}(x)}/{\p_{\qmm}\supparen{i}(x)}$.
    We can now evaluate $H_i(\sigma, \sigma')$ by expanding the probabilities with the Born's rule,
    \begin{align*}
        &H_i(\sigma, \sigma') \\
        &=\expectDistrOf{x\sim\p_{\qmm}\supparen{i}}{\frac{(\p_{\sigma}\supparen{i}(x)-\p_{\qmm}\supparen{i}(x))(\p_{\sigma'}\supparen{i}(x)-\p_{\qmm}\supparen{i}(x))}{(\p_{\qmm}\supparen{i}(x))^2}}\\
        &=\sum_{x}\frac{(\p_{\sigma}\supparen{i}(x)-\p_{\qmm}\supparen{i}(x))(\p_{\sigma'}\supparen{i}(x)-\p_{\qmm}\supparen{i}(x))}{\p_{\qmm}\supparen{i}(x)}\\
        &=\sum_{x}\frac{\Tr[M_x\supparen{i}(\sigma-\qmm)]\Tr[M_x\supparen{i}(\sigma'-\qmm)]}{\Tr[M_x\supparen{i}]/\dims} & \text{(Born's rule)}\\
        &=\dims\sum_{x}\frac{\vvdotprod{\sigma-\qmm}{M_x\supparen{i}}\vvdotprod{M_x\supparen{i}}{\sigma'-\qmm}}{\Tr[M_x]} & \text{(Matrix inner product)}\\
        &=\dims \vadj{\Delta_{\sigma}}\Choi_i\vvec{\Delta_{\sigma'}}.
    \end{align*}
    
    Then, using~\cref{lem:chi-square-expansion}, and the fact that $1+x\le \exp(x)$,
    \begin{align*}
        \chisquare{\expectDistrOf{\sigma\sim\mathcal{D}}{\bfP_{\sigma}}}{\bfP_{\qmm}}&=\expectDistrOf{\sigma, \sigma'}{\prod_{i=1}^n(1+H_i(\sigma, \sigma'))}-1\\
        &\le \;\expectDistrOf{\sigma,\sigma'}{\exp\left\{\sum_{i=1}^nH_i(\sigma, \sigma')\right\}}-1\\
        &=\expectDistrOf{\sigma,\sigma'}{\exp\left\{\dims\sum_{i=1}^\ns\vadj{\Delta_{\sigma}}\Choi_i\vvec{\Delta_{\sigma'}}\right\}}-1\\
        &=\;\expectDistrOf{\sigma, \sigma'}{\exp\{\ns\dims \vadj{\Delta_{\sigma}}\avgChoi_i\vvec{\Delta_{\sigma'}}\}}-1.
    \end{align*}
    The last step follows by linearity and that $\avgChoi=\frac{1}{\ns}\sum_{i=1}^{\ns}\Choi_i$. 
\end{proof}

\paragraph{Explaining the example in~\cref{sec:fixed-disadvantage}.} We now use~\cref{lem:decoupled-chi-square} to explain why choosing a fixed basis measurement $\{\qproj{x}\}_{x=0}^{\dims-1}$ for all copies as in~\cref{sec:fixed-disadvantage} would fail. Since there are only $\dims$ rank-1 projectors, the rank of $\avgChoi$ is $\dims$, but $\avgChoi$ has a dimension of $\dims^2\times\dims^2$ and thus there are a total of $\dims^2-\dims$ eigenvectors with 0 eigenvalues. From~\cref{prop:perturbation-trace-distance}, we know that there must exist a trace-0 $\Delta$ in the 0-eigenspace such that $\sigma=\qmm+\Delta\in\mathcal{P}_{\eps}$. For this particular $\sigma$ the upper bound in~\eqref{equ:chi-square-quantum} is 0, and thus it is impossible to distinguish $\qmm$ and $\sigma$. This is consistent with the discussion in~\cref{sec:fixed-disadvantage}.

We can make a more general argument that to avoid the catastrophic failure similar to the dummy example in~\cref{sec:fixed-disadvantage}, $\avgChoi$ has to be nearly full-rank: $\text{rank}(\avgChoi)\ge (1-o(1))\dims^2$. Thus $(1-o(1))\dims^2$ linearly independent components are needed in all the POVMs. Indeed if otherwise, the dimension of the 0-eigenspace of $\avgLuders$ is $\Omega(\dims^2)$, we can again invoke~\cref{prop:perturbation-trace-distance} (perhaps with some different constants) to find a \emph{single fixed} $\sigma$ that completely fools the measurement scheme.

\subsection{The new lower bound construction}
\label{sec:new-construction}
We now describe a generic construction of a distribution over density matrices that will serve as hard case for both fixed and randomized measurements. {We will take a finite set of trace-0 orthonormal Hermitian matrices. Then, we take a linear combination of these matrices with coefficients chosen at random to be $\pm 1$ (appropriately normalized). When we add maximally mixed state $\qmm$ to this distribution's output we obtain a perturbed distribution around $\qmm$.}

The construction is defined in \cref{def:perturbation}. As discussed in \cref{sec:summary:new-construction}, it can be viewed as a generalization of Paninski's construction \cref{eq:paninski} frequently used for discrete distribution learning and testing.
 \begin{definition}
 \label{def:perturbation}
     Let $\ell\in[\frac{\dims^2}{2}, \dim^2-1]$ and $\hbasis=(V_1, \ldots, V_{\dims^2}=\frac{\eye_\dims}{\sqrt{\dims}})$ be an orthonormal basis of $\Herm{\dims}$, and $\cd$ be a universal constant. Let  $z=(z_1, \ldots, z_\ell)$ be uniformly drawn from $\{-1, 1\}^\ell$,
     \begin{equation}
         \Delta_z = \frac{c\dst}{\sqrt{\dims}}\cdot\frac{1}{\sqrt{\ell}}\sum_{i=1}^\ell z_iV_i, \quad \barDelta_z= \Delta_z\min\left\{1, \frac{1}{\dims \opnorm{\Delta_z}}\right\}.
         \label{equ:delta_z}
     \end{equation}
     Finally we set $\sigma_z=\qmm + \barDelta_z$ whose distribution we denote as $\ptbDistr(\hbasis)$.
       
 \end{definition}
 
The goal of normalizing $\Delta_z$ is to ensure that $\sigma_z$ is a valid density matrix. 
However, after normalization, the trace distance $\tracenorm{\sigma_z-\qmm}$ may not be greater than $\dst$. 
Nevertheless, we can show that the probability of this bad event is negligible. The central claim is ~\cref{thm:rand-mat-opnorm-concentration} which states that the operator norm of a random matrix with independently perturbed orthogonal components is $O(\sqrt{\dims})$ with high probability. The proof is in~\cref{app:prop:perturbation-trace-distance}.
\begin{restatable}{theorem}{randmatopnorm}
\label{thm:rand-mat-opnorm-concentration}
    Let $V_1, \ldots, V_{\dims^2}\in\C^{\dims\times \dims}$ be an orthonormal basis of $\C^{\dims\times \dims}$ and $\ptb_1, \ldots, \ptb_{\dims^2}\in\{-1, 1\}$ be independent symmetric Bernoulli random variables. Let $W=\sum_{i=1}^{\ell}\ptb_iV_i$ where $\ell\le \dims^2$. For all $\alpha>0$, there exists $\cop_\alpha$, {which is increasing in $\alpha$} such that
    \[
    \probaOf{\opnorm{W}>\cop_\alpha\sqrt{\dims}}\le 2\exp\{-\alpha\dims\}.
    \]
\end{restatable}
\begin{remark}
    Standard random matrix theory (e.g. \cite{tao2023topics}[Corollary 2.3.5]) states that if each entry of $W$ is independent and uniform from $\{-1, 1\}$, i.e. $W=\sum_{i,j}z_{ij}E_{ij}$ where $E_{ij}$ is a matrix with 1 at position $(i, j)$ and 0 everywhere else, then $\opnorm{W}=O(\sqrt{\dims})$ with high probability.~\cref{thm:rand-mat-opnorm-concentration} generalizes this argument to arbitrary basis $\{V_i\}_{i=1}^{\dims^2}$. {This could be of independent interest.}
\end{remark}

An immediate corollary of~\cref{thm:rand-mat-opnorm-concentration} is that with appropriately chosen constant $c$ in~\cref{def:perturbation}, $\sigma_z$ is $\dst$ far from $\qmm$ with overwhelming probability, so that we can almost neglect the effect of normalization.
\begin{corollary}
\label{prop:perturbation-trace-distance}
    Let $\dims ^2/2\le \ell\le \dims^2-1$. Let  $z$ be drawn from a uniform distribution over $\{-1,1\}^{\ell}$ , and $\Delta_z, \sigma_z$ are as defined in~\cref{def:perturbation}. Then, there exists a universal constant $\cd\le 10\sqrt{2}$, such that for $\dst<\frac{1}{\cd^2}$, with probability at least $1-2\exp(-\dims)$, $\opnorm{\Delta_z}\le 1/\dims$ and $\tracenorm{\Delta_z}\ge \dst$.   
\end{corollary}
\begin{proof}
By H\"older's inequality, we have that for all matrices $A$,
\[
\opnorm{A}\tracenorm{A}\ge \hsnorm{A}^2.
\]
Note that $\Delta_z=\frac{\cd\dst}{\sqrt{\dims\ell}}W$ and $\hsnorm{\Delta_z}=\frac{c\dst}{\sqrt{\dims}}$. Thus setting $\alpha=1$ and $\cop=\cop_1$ in~\cref{thm:rand-mat-opnorm-concentration}, with probability at least $1-2\exp(-\dims)$,
\[
\opnorm{\Delta_z}\le \frac{\cd\dst}{\sqrt{\dims\ell}}\cdot \kappa \sqrt{\dims}=\frac{c\kappa\dst}{\sqrt{\ell}}.
\]
 This implies that
\[
\tracenorm{\Delta_z}\ge \hsnorm{\Delta_z}^2/\opnorm{\Delta_z}\ge \frac{c\dst}{\kappa}\cdot\frac{\sqrt{\ell}}{\dims}.
\]
In the proof of~\cref{thm:rand-mat-opnorm-concentration} in~\cref{app:prop:perturbation-trace-distance}, we can show that $\kappa=\kappa_1\le 10$. Thus choosing $\cd=\sqrt{2}\kappa\le 10\sqrt{2}$, we guarantee that $\tracenorm{\Delta_z}>\dst$ due to $\ell\ge \dims^2/2$. As long as $\dst\le \frac{1}{200}$, we have $\opnorm{\Delta_z}\le 1/\dims$ and thus $\sigma_z=\qmm + \Delta_z$ is a valid density matrix. This completes the proof of~\cref{prop:perturbation-trace-distance}.
\end{proof}

\subsection{Bounding the min-max and max-min divergences}

To complete the lower bound proof, we apply \cref{lem:decoupled-chi-square} to our new lower bound construction in \cref{def:perturbation}. 

Different bounds for min-max and max-min divergences in~\cref{lem:min-max-and-max-min} are due to whether or not nature can choose the basis $\hbasis$ dependent on $\avgLuders$, which in turn depends on the measurements $\POVM^\ns$. For randomized schemes, we need to upper bound the min-max divergence, and we can simply choose a fixed $\hbasis$ that is uniformly bad for all $\POVM^\ns$. For fixed measurements, however, under the max-min framework, nature could choose the hard distribution depending on $\POVM^\ns$. Specifically, with $\hbasis=\hbasis_{\avgLuders}$ and $\ell$ small, $\sigma_z-\qmm$ completely lies in an eigenspace of $\avgLuders$ with the $\ell$ smallest eigenvalues, thus generalizing the intuition from the toy example in~\cref{sec:toy-example}.

\cref{thm:chi-square-upper-bound} upper bounds the chi-square divergence in~\eqref{equ:chi-square-quantum} for the construction in~\cref{def:perturbation}. We obtain different bounds depending on the choice of $\hbasis$. Specifically, if $\hbasis$ is the eigenbasis of $\avgLuders$ where $V_1, \ldots, V_\ell$ have the smallest eigenvalues, the upper bound would be much tighter. 

\begin{theorem}
\label{thm:chi-square-upper-bound}
    Let $\hbasis=(V_1, \ldots, V_{\dims^2}=\frac{\eye_\dims}{\sqrt{\dims}})$ be an orthonormal basis of $\Herm{\dims}$, and $\povmset$ be the set of allowed POVMs. Let $\sigma,\sigma'\sim \ptbDistr(\hbasis)$ and $\Delta_\sigma=\sigma-\qmm$. Then for $\ns\le 
    {\frac{\dims^2}{6\cd^2\eps^2\max_{\POVM\in\povmset}\hsnorm{\Luders_{\POVM}}}}$,
    \begin{equation}
          \log\expectDistrOf{\sigma,\sigma'}{\exp\left\{\ns\dims \cdot\vadj{\Delta_\sigma}\avgChoi\vvec{\Delta_{\sigma'}} \right\}}= \bigO{\frac{\ns^2\eps^4}{\dims^4}\max_{\POVM\in\povmset}\hsnorm{\Luders_{\POVM}}^2 }. 
          \label{equ:chi-square-random}
    \end{equation}
    If $\hbasis$ is the eigenbasis of $\avgLuders$ with eigenvalues $\lambda(V_1)\le \cdots\le \lambda(V_{\dims^2})$, then for $\ns\le {\frac{\dims^3}{9\cd^2\eps^2\max_{\POVM\in\povmset}\tracenorm{\Luders_{\POVM}}}}$,
    \begin{equation}
    \label{equ:chi-square-fixed}
           \log\expectDistrOf{\sigma,\sigma'}{\exp\left\{\ns\dims \cdot\vadj{\Delta_\sigma}\avgChoi\vvec{\Delta_{\sigma'}} \right\}}=\bigO{\frac{\ns^2\eps^4}{\dims^6}\max_{\POVM\in\povmset}\tracenorm{\Luders_{\POVM}}^2 }. 
    \end{equation}
\end{theorem}

We defer the proof to~\cref{app:thm:chi-square-upper-bound} and proceed to prove the main result \cref{thm:lower-channel}. For randomized measurements, we choose $\hbasis$ to be an arbitrary basis that satisfies~\cref{def:perturbation} (e.g. the generalized Gell-Mann basis, or the Pauli basis). Then combining~\eqref{equ:min-max-lb}~\eqref{equ:chi-square-quantum}~\eqref{equ:chi-square-random}, we must have
\begin{align*}
    \frac{2}{25}\le \exp\left\{\bigO{\frac{\ns^2\eps^4}{\dims^4}\max_{\POVM\in\povmset}\hsnorm{\Luders_{\POVM}}^2 }\right\}-1.
\end{align*}
Therefore $\ns = \bigOmega{\frac{\dims^2}{\eps^2\max_{\POVM\in \povmset}\hsnorm{\Luders_{\POVM}}}}$ for randomized measurements. 

To bound the max-min divergence~\eqref{equ:max-min-lb} for fixed measurements, for all fixed $\POVM^n$ we choose we $\hbasis$ as stated in~\cref{thm:chi-square-upper-bound}. Combining~\eqref{equ:max-min-lb}~\eqref{equ:chi-square-quantum}~\eqref{equ:chi-square-fixed},
\begin{align*}
    \frac{2}{25}\le \exp\left\{\bigO{\frac{\ns^2\eps^4}{\dims^6}\max_{\POVM\in\povmset}\tracenorm{\Luders_{\POVM}}^2 }\right\}-1.
\end{align*}
Rearranging the terms we get $\ns = \bigOmega{\frac{\dims^3}{\eps^2\max_{\POVM\in \povmset}\tracenorm{\Luders_{\POVM}}}}$ for fixed measurements. This completes the proof of \cref{thm:lower-channel}.

\section{Algorithm for randomized $\ab$-outcome measurements}
\label{sec:rand-k-upper}
Recall that $\dims=2^{\nqubits}$. We further assume that $\ab\le\dims$ is a power of 2 without loss of generality\footnote{If not, we can round $\ab$ down to the nearest power of 2, and the copy complexity only differs by a constant.}, so that $\dims/\ab$ is an integer. We use $\Haar{\dims}$ to denote the Haar measure over $\dims\times\dims$ unitary matrices.

The idea of the algorithm is to sample $\ab$ orthogonal projections $\Proj_1, \ldots, \Proj_{\ab}$ uniformly from the Haar measure with $\rank(\Proj_i)=r\eqdef \dims/\ab$. Then we apply POVM $\POVM\eqdef\{\Proj_j\}_{j=1}^{\ab}$ to all copies. Finally, we apply classical identity testing to the outcomes $\bx$.  Details are described in~\cref{alg:rand-k}. It generalizes the algorithm in~\cite{BubeckC020} to $\ab<\dims$. 
\begin{algorithm}
\caption{State certification with randomized $\ab$-outcome measurements}
\label{alg:rand-k}
    \begin{algorithmic}
        \State \textbf{Input:} $\ns$ copies of an unknown state $\rho$, state description $\qkn$, desired accuracy $\eps$.
        \State \textbf{Output:} {\isthestate} if $\rho=\qkn$, {\notthestate} if $\tracenorm{\rho-\qkn}>\eps$.
        \State $r\eqdef\dims/\ab$.
        \State Sample a unitary matrix $U=[\qbit{u_1}, \ldots,\qbit{u_\dims}]$ from the Haar measure $\Haar{\dims}$.
        \State Define a POVM $\POVM_U\eqdef\{\Proj_x^U\}_{x=1}^{\ab}$ where $\Proj_x^{U}=\sum_{i=r(x-1)+1}^{rx}\qproj{u_i}$. 
        \State Let $\p_\rho^U$ be the outcome distribution when applying $\POVM_U$ to $\rho$.
        \State Apply $\POVM_U$ to all copies and obtain outcomes $\bx=(x_1, \ldots, x_{\ns})$.
        \State \Return{TestIdentityL2$(\p_{\qkn}^U, \bx, 0.07\eps/\dims, \delta=0.01)$.}
    \end{algorithmic}
\end{algorithm}

The main step is to upper bound $\normtwo{\p_{\rho}^U}$ and lower bound $\normtwo{\p_{\rho}^U-\p_{\qkn}^U}$ with at least constant probability over $U$. We state the result in~\cref{lem:l2-norm-haar}, which is similar to domain compression for communication-constrained classical distribution testing \cite[Theorem VI.2]{ACT:19:IT2}. The lemma states that under random unitary projection, the $\ell_2$ distance $\normtwo{\p_{\rho}^U-\p_{\qkn}^U}$ roughly remains the same with $\ab$, but the $\ell_2$ norm  $\normtwo{\p_{\rho}^U}$ becomes smaller with larger $\ab$. The proof is in~\cref{app:lem:l2-norm-haar} and involves computing moments of the Haar measure using Weingarten calculus.
    \begin{lemma}[Quantum domain compression]
    \label{lem:l2-norm-haar}
        Let $U\sim \Haar{\dims}$, $\POVM_U\eqdef\{\Proj_x^U\}_{x=1}^{\ab}$ where $\Proj_x^{U}=\sum_{i=r(x-1)+1}^{rx}\qproj{u_i}$ and $\p_\rho^U$ be the outcome distribution when applying $\POVM_U$ to $\rho$. Then, for all states $\rho,\sigma$,
        \[
        \probaOf{\normtwo{\p_{\rho}^U}\le \frac{10}{\sqrt{\ab}}}\ge 0.98, \quad \probaOf{\normtwo{\p_{\rho}^U-\p_{\sigma}^U}\ge \frac{0.07}{\sqrt{\dims}}\hsnorm{\rho-\sigma}}\ge 0.13. 
        \]
    \end{lemma}
    
Using union bound and \cref{thm:distr-testing}, \cref{lem:l2-norm-haar} only ensures that \cref{alg:rand-k} outputs the correct answer with probability 0.1 with $\ns = \bigO{\frac{\dims^2}{\eps^2\sqrt{\ab}}}$ copies. We can boost the probability using \cref{alg:rand-k-boost}, which is a standard amplification algorithm. \cref{thm:rand-k-upper} gives its guarantee and completes the upper bound in~\cref{thm:random-k}. 
\begin{algorithm}
\caption{State certification: amplification algorithm}
\label{alg:rand-k-boost}
    \begin{algorithmic}
        \State \textbf{Input:} $\ns$ copies of an unknown state $\rho$, state description $\qkn$, desired accuracy $\eps$.
        \State \textbf{Output:} {\isthestate} if $\rho=\qkn$, {\notthestate} if $\tracenorm{\rho-\qkn}>\eps$.
        \State Divide the copies into $T$ groups of equal size.
        \State For group $i$, run \cref{alg:rand-k} and obtain output $y_i\in \{\isthestate, \notthestate\}$.
        \State Let $b=\frac{1}{T}\sum_{i=1}^T\indic{y_i=\notthestate}, t_1=0.03, t_2=0.1$.
        \If{$b>(t_1+t_2)/2$}
        \Return{\notthestate}
        \Else~\Return{\isthestate}
        \EndIf
    \end{algorithmic}
\end{algorithm}

\begin{theorem}
\label{thm:rand-k-upper}
    Let $\dims\ge 100$ and $2\le \ab\le \dims$,~\cref{alg:rand-k-boost} can test whether $\rho=\qkn$ or $\tracenorm{\rho-\qkn}>\eps$ with probability at least $2/3$ using $\ns=\bigO{\frac{\dims^2}{\eps^2\sqrt{\ab}}}$ copies of $\rho$. 
\end{theorem}
\begin{proof}
     First, consider the case when $\rho=\qkn$. Then we always have $\p_{\rho}^U=\p_{\qkn}^U$. Substituting the $\ell_2$ norms and $\delta=0.01$ , \cref{thm:distr-testing} guarantees that with $\ns/T=\bigO{\frac{1}{\sqrt{\ab}}\cdot \frac{1}{(\eps/\dims)^2}}=\bigO{\frac{\dims^2}{\eps^2\sqrt{\ab}}}$,
    \[
    \probaOf{y_i=\notthestate\bigg|\p_{\rho}^U\le \frac{10}{\sqrt{\ab}}}\le \delta=0.01.
    \]
    Using the first part of \cref{lem:l2-norm-haar}, when $\rho=\qkn$, we have
    \[
    \probaOf{y_i=\notthestate}\le \probaOf{\normtwo{\p_{\rho}^U}>\frac{10}{\sqrt{\ab}}} + \probaOf{y_i=\notthestate\bigg|\p_{\rho}^U\le \frac{10}{\sqrt{\ab}}}\probaOf{\normtwo{\p_{\rho}^U}\le \frac{10}{\sqrt{\ab}}}\le 0.01+0.02=0.03
    \]

    Then consider the case when $\tracenorm{\rho-\qkn}\ge \eps$. By Cauchy-Schwarz 
    $$\hsnorm{\rho-\qkn}\ge\eps/\sqrt{\dims},$$
    and thus $\normtwo{\p_{\rho}^U-\p_{\sigma}^U}\ge 0.07\eps/\dims$ with probability at least 0.13 by \cref{lem:l2-norm-haar}. 
    
    By union bound, both $\normtwo{\p_{\rho}^U}\le \frac{10}{\sqrt{\ab}}$ and  $\normtwo{\p_{\rho}^U-\p_{\sigma}^U}\ge 0.07\eps/\dims$ with probability at least 0.13-(1-0.98)=0.11.  
    We conclude the algorithm TestIdentityL2 outputs the correct answer $\notthestate$ using $\ns/T = \bigO{\frac{\dims^2}{\eps^2\sqrt{\ab}}}$ copies with probability at least,
    \[
    \probaOf{y_i=\notthestate}\ge 0.11(1-\delta)>0.1.
    \]
    
    Thus, our goal is to distinguish between $\bernoulli{p_1}$ and $\bernoulli{p_2}$ where $p_1\le t_1\eqdef 0.03$ and $p_2\ge t_2\eqdef 0.1$ with constant $T$ number of samples. We use the Hoeffding's inequality,
    \begin{lemma}[Hoeffding]
    \label{lem:hoeffding}
        Let $X_1, \ldots, X_n$ be independent random variables with $a_i\le X_i\le b_i$ and $X=\sum_{i=1}^nX_i$. Let $\expect{X}=\mu$. Then,
        \[
        \probaOf{X-\mu>t}\le \exp\left\{-\frac{2t^2}{\sum_{i=1}^n(b_i-a_i)^2}\right\}.
        \]
    \end{lemma}
    By symmetry, the same bound holds for the left tail probability $\probaOf{X-\mu<-t}$.
    Thus when $T=\frac{2}{(t_2-t_1)^2}\log(\frac{1}{\alpha})$, for both $i=1,2$, we have $|b-p_i|>(t_2-t_1)/2$ with probability at most $\alpha$. Setting $\alpha=2/3$, $T$ is a constant and thus $\ns = \bigO{\frac{\dims^2}{\eps^2\sqrt{\ab}}}$ as desired.
\end{proof}
\begin{remark}
    The proof of \cref{thm:rand-k-upper} only uses the 4th moment of the Haar measure. Thus, we can sample $U$ from (approximate) unitary 4-designs instead of the Haar measure. 
\end{remark}
\begin{remark}
    \cref{lem:l2-norm-haar} does not require $\qkn$ to be known. In this case, with $\ns$ copies of $\qkn$, we can adapt~\cref{alg:rand-k} to closeness testing by measuring all copies of $\qkn$ and applying classical closeness testing algorithm~\cite[Lemma 2.3]{DiakonikolasK16} in the final step. Thus, \cref{alg:rand-k-boost} naturally extends to closeness testing with the same copy complexity of $\ns=O(\dims^2/(\sqrt{\ab}\eps^2))$.
\end{remark}

\section{Algorithm for fixed Pauli measurements}
\label{sec:fixed-Pauli-upper}
\subsection{Pauli observables}
For an $\nqubits$-qubit system (and thus $\dims=2^{\nqubits}$), the set of Pauli observables is $\pauliObsSet\eqdef \Sigma^{\otimes\nqubits}\setminus\{\eye_\dims\}$, where $\Sigma\eqdef \{\pauliI, \pauliX, \pauliY, \pauliZ\}$ are the Pauli operators,
\[
\pauliI\eqdef\eye_2 = \begin{bmatrix}
    1 & 0 \\
    0 & 1
\end{bmatrix},\quad
\pauliX = \begin{bmatrix}
    0 & 1 \\
    1 & 0
\end{bmatrix},\quad
\pauliY = \begin{bmatrix}
    0 & i \\
    -i & 0
\end{bmatrix},\quad
\pauliZ = \begin{bmatrix}
    1 & 0\\
    0 & -1
\end{bmatrix}.
\]
$\pauliObsSet$ consists of $\dims^2-1$ matrices. We have the standard fact about Pauli observables,
\begin{fact}
\label{fact:pauli}
    Let $P, Q\in \mathcal{P}$ be two Pauli observables. Then, 
\[
P^2=P, \;\Tr[P]=0,\; \Tr[PQ]=\hdotprod{P}{Q}=\dims\indic{P=Q}.
\]
Therefore, together with $\eye_\dims$, the set $\Sigma^{\otimes N}$ forms an orthogonal basis for $\Herm{\dims}$. 
\end{fact}

Each $P\in\pauliObsSet$ defines a 2-outcome POVM,
\[
\POVM_P=\{M_0^P, M_1^P\},\; M_0^P=\frac{\eye_\dims - P}{2}, M_1^P=\frac{\eye_\dims+P}{2}.
\]
For a state $\rho$, we define 
\[
\p_{\rho}\eqdef[\prPauli{P}{\rho}: P\in \pauliObsSet]\in[0, 1]^{\dims^2-1}, \quad \prPauli{P}{\rho}=\Tr[\rho M_1^P].
\]
which is the collection of probabilities of seeing ``1'' after applying $\POVM_P$. 
\cref{lem:prob-vec-and-pauli} relates $\normtwo{\p_{\rho}-\p_{\sigma}}$ to $\hsnorm{\rho-\sigma}$ for two arbitrary states $\rho,\sigma$. 
\begin{lemma}
\label{lem:prob-vec-and-pauli}
    Let $\rho$ and $\sigma$ be two quantum states. Then, $\normtwo{\p_{\rho}-\p_{\sigma}}=\frac{\sqrt{\dims}}{2}\hsnorm{\rho-\sigma}$ .
\end{lemma} 

\begin{proof}
The proof uses that Pauli observables and $\eye_\dims$ forms a basis of $\Herm{\dims}$.Thus we can represent a state $\rho$ as,
\begin{equation*}
    \rho = \frac{\eye_\dims}{\dims}+\sum_{P\in \pauliObsSet}\frac{\Tr[\rho P]P}{\dims}.
\end{equation*}
    Therefore, by the definition of Hilbert-Schmidt norm and that $\Tr[PQ]=\dims\cdot \delta_{P,Q}$,
    \begin{align}
        \hsnorm{\rho-\sigma}^2&=\hsnorm{\sum_{P\in\pauliObsSet}\frac{1}{\dims}(\Tr[P\rho]-\Tr[P\sigma])P}^2=\frac{1}{\dims}\sum_{P\in\pauliObsSet}(\Tr[P\rho]-\Tr[P\sigma])^2.
        \label{equ:pauli-state-l2}
    \end{align}
    
We can compute $\prPauli{P}{\rho}=\Tr[\rho M_1^P]$ using Born's rule,
\begin{equation*}
    \prPauli{P}{\rho}=\frac{1+\Tr[\rho P]}{2}.
\end{equation*}
Therefore,
    \begin{equation}
        \normtwo{\p_{\rho}-\p_{\sigma}}^2=\frac{1}{4}\sum_{P\in\pauliObsSet}(\Tr[P\rho]-\Tr[P\sigma])^2.
    \label{equ:pauli-probability-l2}
    \end{equation}
    Comparing~\eqref{equ:pauli-state-l2} and~\eqref{equ:pauli-probability-l2} proves the lemma.
\end{proof}

\subsection{Algorithm}
The idea is to reduce the problem to testing product Bernoulli distributions defined as follows.
\begin{definition}[Product Bernoulli]
    Given $\p=(p_1, \ldots, p_\dims)\in[0, 1]^{\dims}$, $\bernoulli{\p}$ is the distribution of $\bx=(x_1, \ldots, x_\dims)\in\{0,1\}^{\dims}$ where $x_i\sim \bernoulli{p_i}$ are independent Bernoulli distributions. 
\end{definition}
When we apply all Pauli observables $P\in\pauliObsSet$ to $\dims^2-1$ copies, the $\dims^2-1$ outcomes simulate one i.i.d. sample from $\bernoulli{\p_{\rho}}$. Thus, we can apply a classical product Bernoulli testing algorithm, whose performance under $\ell_2$ distance is stated in \cref{thm:prod-bern-testing-l2}. 
\begin{theorem}[{\citet[Theorem 6]{Canonne2020testing}}]
\label{thm:prod-bern-testing-l2}
    Let $\p,\q\in[0, 1]^{\dims}$ where $\p$ is unknown and $\q$ is known, and $\bx^{\ns}$ be $\ns$ i.i.d. samples from $\bernoulli{\p}$. There exists an algorithm TestProdBernL2$(\q, \bx^{\ns},\dims,\eps)$ that outputs $\isthestate$ if $\p=\q$ and $\notthestate$ if $\normtwo{\p-\q}>\eps$ w.p. at least 2/3 using $\ns=O(\sqrt{\dims}/\eps^2)$ samples.
\end{theorem}

Details are in~\cref{alg:fixed-pauli}, which is similar to the simulate-and-infer in~\citep{acharya2020distributed}. \cref{thm:fixed-pauli-upper} proves its performance. Combining with~\cref{thm:fixed-k-lower} for $\ab=2$, we prove the tight copy complexity of $\Theta(\dims^3/\eps^2)$ in~\cref{thm:fixed-pauli-upper-lower} for Pauli observables and 2-outcome measurements.
\begin{algorithm}
\caption{Quantum state certification with fixed Pauli observables}
    \begin{algorithmic}
        \State \textbf{Input}: $\ns$ copies of $\rho$, known state description $\qkn$.
        \State \textbf{Output}: {\isthestate} if $\rho=\qkn$, {\notthestate}  if $\tracenorm{\rho-\qkn}>\eps$.
        \State Let $\q=\bernoulli{\p_{\qkn}}$.
        \State Divide $\ns$ copies into $L\eqdef\lfloor\ns/(\dims^2-1))\rfloor$ groups, each with size $\dims^2-1$. Excess copies are omitted.
        \State Group $j$ apply all $P\in\pauliObsSet$ to each copy and obtain $\dims^2-1$ outcomes $\bx_j=(x_{j,P})_{P\in\pauliObsSet}$.
        \State\Return {TestProdBernL2$(\q, \bx^L=(\bx_1,\ldots, \bx_L),\dims^2-1, \eps/2)$}.
    \end{algorithmic}
    \label{alg:fixed-pauli}
\end{algorithm}

\begin{theorem}
\label{thm:fixed-pauli-upper}
    ~\cref{alg:fixed-pauli} can test whether $\rho=\qkn$ or $\tracenorm{\rho-\qkn}\ge \eps$ with probability at least 2/3 using $\ns = O(\dims^3/\eps^2)$ copies of $\rho$.
\end{theorem}
\begin{proof} 
If $\rho=\qkn$, then $\p_{\rho}=\p_{\qkn}$. If $\tracenorm{\rho-\qkn}\ge \eps$, by Cauchy-Schwarz $\hsnorm{\rho-\qkn}\ge \frac{\eps}{\sqrt{\dims}}$, and thus by~\cref{lem:prob-vec-and-pauli},
$$
\|\p_{\rho}-\p_{\qkn}\|_2\ge \eps/2.
$$
Setting $\dims\leftarrow \dims^2-1,\eps\leftarrow\eps/2$ in~\cref{thm:prod-bern-testing-l2}, $
L = \bigO{\frac{\sqrt{\dims^2-1}}{\eps^2}}=\bigO{\frac{\dims}{\eps^2}}$ samples from $\bernoulli{\p_{\rho}}$ are sufficient to test whether $\p_{\rho}=\p_{\qkn}$ or $\|\p_{\rho}-\p_{\qkn}\|_2\ge \eps/2$, which implies $\ns\le(\dims^2-1)(L+1)=O(\dims^3/\eps^2)$, exactly as desired. 
\end{proof}

\section{Algorithm for $\ab$-outcome fixed measurements}
\label{sec:fixed-d-outcome}
Similar to \cref{sec:rand-k-upper}, we assume that both $\ab$ and $\dims$ are powers of 2. We begin with preliminaries for quantum designs and mutually unbiased bases. Then we propose an algorithm for $\ab=\dims$ which makes minor modifications to the algorithm in \cite{Yu21sample} to satisfy the constraint on the number of outcomes. 
\subsection{Quantum designs and mutually unbiased basis}
Our algorithm uses quantum 2-designs. A $t$-design preserves the statistics of the Haar measure up to $t$-th order moments. Below is the formal definition.
\begin{definition}[Quantum $t$-design]
\label{def:spherical-t-design}
    Let $t$ be a positive integer, we say that a finite set of normalized vectors $\{\qbit{\psi_x}\}_{x=1}^m$ in $\C^\dims$ is a \textit{quantum $t$-design} if
    \[
    \frac{1}{m}\sum_{x=1}^m \qproj{\psi_x}^{\otimes t}=\int \qproj{\psi}^{\otimes t} d \mu(\psi),
    \]
    where $\mu$ is the Haar measure on the unit sphere in $\C^\dims$.
\end{definition}

Maximal mutually unbiased bases (MUB) is a well-known example of 2-design that exists for $\dims$ which are prime powers. We refer the readers to \cite{durt2010mutually} for a survey.
\begin{theorem}[\cite{klappenecker2005mutually}]
    Let $\dims$ be a prime power, then there exists a maximal MUB, i.e. $\dims+1$ orthonormal bases $\{\qbit{\psi_{x}^l}\}_{x=1}^{\dims}, l=1, \ldots, \dims+1$ such that the collection of all vectors $\{\qbit{\psi_{x}^l}\}_{x, l}$ is a $2$-design.
    \label{thm:mub-2-design}
\end{theorem}

\subsection{Algorithm for $\ab=\dims$}
We first review the algorithm in~\cite{Yu21sample}. It treats the collection of all $\dims+1$ bases as one single POVM, 
\[
\POVM_{MUB}=\left\{\frac{1}{\dims+1}\qproj{\psi_x^l}\right\}_{x\in[\dims], l\in[\dims+1]},
\]
which has $\dims(\dims+1)$ outcomes. This measurement is applied to all copies. We denote the outcome distribution for the state $\rho$ as $\p_{\rho}\in \R^{\dims(\dims+1)}$. Maximal MUB ensures that for $\rho$ far from $\qkn$, the outcome distributions $\p_{\rho}$, $\p_{\qkn}$ should also be far.
\begin{lemma}[{\cite[Lemma 2]{Yu21sample}}] Let $\p_{\rho}$ be the distribution when applying a maximal MUB to $\rho$. Then
\label{lem:mub-l2-distance}
     \[
    \normtwo{\p_{\rho}}\le\frac{\sqrt{2}}{\dims+1},\quad \normtwo{\p_{\rho}-\p_{\qkn}}=\frac{\hsnorm{\rho-\qkn}}{\dims+1}.
    \]
    In particular if $\tracenorm{\rho-\qkn}>\eps$, due to Cauchy-Schwarz $\normtwo{\p_{\rho}-\p_{\qkn}}=\frac{\hsnorm{\rho-\qkn}}{\dims+1}\ge \frac{\eps}{(\dims+1)\sqrt{\dims}}$.
\end{lemma}
Thus we can apply the distribution testing algorithm in \cref{thm:distr-testing} to distinguish between $\p_{\rho}$ and $\p_{\qkn}$. Setting $b\leftarrow\frac{\sqrt{2}}{\dims+1}$ and $\eps\leftarrow \frac{\eps}{(\dims+1)\sqrt{\dims}}$, to achieve success probability $\delta$, the number of copies required is
\begin{equation}
    \ns = 1000\frac{\sqrt{2}}{\dims+1}\cdot\frac{\dims(\dims+1)^2}{\eps^2}\log\frac{1}{\delta}=1000\sqrt{2}\cdot \frac{\dims(\dims+1)}{\eps^2}\cdot\log\frac{1}{\delta}.
    \label{equ:}
\end{equation}

Note that $\POVM_{MUB}$ can be viewed as uniformly drawing $l$ from $[d+1]$ and apply $\POVM_l=\{\qproj{\psi_{x}^l}\}_{x=1}^{\dims}$ for each copy. We can thus parameterize the outcome by $(X, L)\in [\dims]\times[\dims+1]$. In particular, when $(X, L)\sim\p_{\rho}$, we have $\probaOf{L=l}=\frac{1}{\dims+1}$.

Using this observation, to design an algorithm with $\dims$ outcomes per measurement, instead of drawing the $\dims+1$ basis measurements uniformly, we divide all copies into $\dims+1$ equal groups and then apply a basis measurement $\POVM_l=\{\qproj{\psi_{x}^l}\}_{x=1}^{\dims}$ for the $l$-th group. Thus we always have the same number of samples for each $l\in[\dims+1]$. We can still apply the identity testing algorithm in \cref{thm:distr-testing} via a simple reduction. The full algorithm is stated in \cref{alg:no-shared-d}.

\begin{algorithm}
\caption{Fixed measurement state certification $\ab=\dims$}
\label{alg:no-shared-d}
    \begin{algorithmic}
        \State \textbf{Input}: $\ns$ copies of state $\rho$.
        \State \textbf{Output} YES if $\rho=\qkn$, NO otherwise.
        \State Let $\{\qbit{\psi_x^l}\}_{x=1}^d, l=1, \ldots, \dims+1$ be a maximal MUB.
        \State Divide the copies into $\dims+1$ equally-sized groups, each group has $\ngr=\ns/(\dims+1)$ copies.
        \State For group $l$, apply basis measurement $\POVM_l = \{\qproj{\psi_x^l}\}_{x=1}^{\dims}$. Let the outcomes be $x_1\supparen{l}, \ldots, x_{\ngr}\supparen{l}$.
        \State Generate $\ns/2$ i.i.d. samples from the uniform distribution over $[\dims+1]$, and for $l\in[\dims+1]$ let $\nspu_l$ be the number of times that $l$ appears.
        \State Let $\bx=(\bx_1, \ldots, \bx_{\dims+1})$ where $\bx_l=(x_1\supparen{l}, \ldots, x\supparen{l}_{\min\{\ngr, m_l\}})$.
        \State \Return TestIdentityL2$\left(\p_{\qkn}, \bx, \frac{\eps}{(\dims+1)\sqrt{\dims}}, \delta=1/6\right)$.
    \end{algorithmic}
\end{algorithm}

\begin{theorem}
\label{thm:fixed-d}
    Given $\ns$ copies from $\rho$, \cref{alg:no-shared-d} can test whether $\rho=\qkn$ or $\tracenorm{\rho-\qkn}>\eps$ with $\ns=O(\dims^2/\eps^2)$ copies.
\end{theorem}
\begin{proof}
    We just need to ensure that with high probability, there are sufficient samples to run TestIdentityL2, i.e. for all $l$, $\ngr\ge m_l$. We use the Chernoff bound,
    \begin{lemma}[Multiplicative Chernoff bound]
    \label{lem:chernoff}
        Let $X_1, \ldots, X_n$ be i.i.d. with $\expect{X_i}=\mu$. Then,
        \begin{align*}
            &\probaOf{\sum_{i=1}^nX_i\ge n(1+\alpha)\mu}\le \exp\left\{-\frac{n\alpha^2\mu}{2+\alpha}\right\}, \;\alpha>0,\\
            &\probaOf{\sum_{i=1}^nX_i\le  n(1-\alpha)\mu}\le \exp\left\{-\frac{n\alpha^2\mu}{2}\right\},\;\alpha\in(0, 1).
        \end{align*}
    \end{lemma}
    Note that $m_l\sim \binomial{\frac{\ns}{2}}{\frac{1}{\dims+1}}$. Thus setting $\alpha=1$ and $\mu=\frac{\ns}{2(\dims+1)}$ in \cref{lem:chernoff},
    \[
    \probaOf{m_l>\ngr}=\probaOf{m_l>\frac{\ns}{\dims+1}}\le \exp\left\{-\frac{\ns}{6(\dims+1)}\right\}.
    \]
    By union bound, 
    \[
    \probaOf{\exists\, l, m_l>\ngr}\le (\dims+1)\exp\left\{-\frac{\ns}{6(\dims+1)}\right\}.
    \]
    Thus probability is at most $\delta'=\frac{1}{6}$ if $\ns>6(\dims+1)\ln(6(\dims+1))$. 

    From \cref{thm:distr-testing}, we can use $\ns=2600\dims(\dims+1)/\eps^2$ copies to distinguish between $\p_{\rho}$ and $\p_{\qkn}$ with probability at least $\frac{5}{6}$. The number of copies is large enough to ensure that $ \probaOf{\exists\, l, m_l>\ngr}\le \frac{1}{6}$. Thus, the probability that all $m_l<\ngr$ \emph{and} TestIdentityL2 returns the correct output is at least $\frac{5}{6}-\delta'=\frac{2}{3}$, exactly as desired.
\end{proof}

\subsection{Algorithm for $\ab<\dims$}

When $\ab<\dims$, we need to simulate the output of each $\POVM_l$ using multiple copies of $\rho$. We use existing results from distributed simulation of discrete distributions under classical information constraints. 

\begin{definition}[$\simprob$-simulation]
    In the problem of $\simprob$-simulation, we have $\ns$ players and an unknown distribution $\p$ over $[\dims]$. Each player receives an i.i.d. sample from $\p$ and can only send $\ell$ bits to a central server. The server then tries to generate $[\hat{X}]\in[\dims]\cup \{\bot\}$ that simulates $\p$ where
\[
\probaOf{\hat{X}=x|\hat{X}\ne\bot}=\p_x,\quad \probaOf{\hat{X}=\bot}\le \simprob.
\]
We say $\hat{X}$ is a successful simulation if and only if $\hat{X}\ne\bot$.
\end{definition}

We use the protocol introduced by \cite{ACT:19:IT2} that simulates an i.i.d. sample from an arbitrary $\dims$-ary distribution. 

\begin{theorem}[{\cite[Theorem IV.5]{ACT:19:IT2}}]
\label{thm:simulation}
    For every $\simprob\in(0, 1)$, there exists an algorithm that $\simprob$-simulates $\dims$-ary distributions using
    \[
    M=40\left\lceil\log\frac1\simprob\right\rceil\left\lceil\frac{\dims}{2^\ell-1}\right\rceil
    \]
    players, where each player receives one i.i.d. sample and can only output $\ell=\log \ab$ bits. The algorithm is deterministic for each player and only requires private randomness.
\end{theorem}

Thus, we use roughly $O(\dims/\ab)$ copies to simulate the outcome distribution of applying each basis measurement $\POVM_l$ to one copy of $\rho$. This results in a $O(\dims/\ab)$ factor increase in copy complexity compared to the case when the number of outcomes is at least $\dims$. Details are described in \cref{alg:no-shared-k}. Formal copy complexity guarantee of \cref{alg:no-shared-k} is stated in \cref{thm:fixed-k}

\begin{algorithm}
\caption{Fixed measurement state certification $\ab<\dims$}
\label{alg:no-shared-k}
    \begin{algorithmic}
        \State \textbf{Input}: $\ns$ copies of state $\rho$.
        \State \textbf{Output} YES if $\rho=\qkn$, NO otherwise.
        \State Let $\{\qbit{\psi_x^l}\}_{x=1}^d, l=1, \ldots, \dims+1$ be a maximal MUB.
        \State Divide the copies into $\dims+1$ equally-sized groups, each group has $\ngr=\ns/(\dims+1)$ copies.
        \For{$l=1, \ldots, \dims+1$}
        \State For group $l$, apply basis measurement $\POVM_l = \{\qproj{\psi_x^l}\}_{x=1}^{\dims}$. Let the outcomes be $x_1\supparen{l}, \ldots, x_{\ngr}\supparen{l}$.
        \State Further divide into groups of size $M=40\left\lceil\log\frac1\simprob\right\rceil\left\lceil\frac{\dims}{\ab-1}\right\rceil$ where $\simprob=0.01$.
        \State Apply the  $\ell$-bit simulation protocol in \cref{thm:simulation} with $\ell=\log \ab$. Obtain simulated outputs $\hat{x}_1\supparen{l}, \ldots, \hat{x}_{\ngr/M}\supparen{l},$.
        \State Let $(\tilde{x}_1, \ldots, \tilde{x}_{n_l})$ be the set of successful simulations (i.e. those that are not $\bot$) where $n_l$ is the number of successful simuations.
        \EndFor
        \State Generate $\ns/(2M)$ i.i.d. samples from the uniform distribution over $[\dims+1]$, and for $l\in[\dims+1]$ let $\nspu_l$ be the number of times that $l$ appears.
        \State Let $\bx_l=(\tilde{x}\supparen{l}_1, \ldots, \tilde{x}\supparen{l}_{\min\{n_l, \nspu_l\}})$ which  keeps at most $\nspu_l$ successful simulations from each group.
        \State Let $\bx=(\bx_1, \ldots, \bx_{\dims+1})$.
        \State \Return TestIdentityL2$\left(\p_{\qkn}, \bx, \frac{\eps}{(\dims+1)\sqrt{\dims}}, \delta=1/6\right)$.
    \end{algorithmic}
\end{algorithm}

\begin{theorem}
\label{thm:fixed-k}
    For $\ab<\dims$, \cref{alg:no-shared-k} can test whether $\rho=\qkn$ or $\tracenorm{\rho-\qkn}>\eps$ with probability at least $2/3$ using $\ns=O(\frac{\dims^3}{\ab\eps^2})$ copies.
\end{theorem}
\begin{proof}
The proof is very similar to \cref{thm:fixed-d} where we argue that with high probability there are enough samples to run TestIdentityL2, or $m_l<n_l$ for all group $l$. The only difference is that for each group $l$ the number of samples $n_l$ is random.  

For convenience we set $\hat{\ns}=\ns/M$ and $\hat{\ns}_0=\ngr/M$. Note that

$$m_l>n_l\implies m_l>\frac{3}{4}\hat{\ns}_0 \text{ or } n_l<\frac{3}{4}\hat{\ns}_0. $$
Thus it suffices to show that both events occur with very small probability. Since $m_l\sim\binomial{\frac{\hat{\ns}}{2}}{\frac{1}{\dims+1}}$ and $\expect{m_l}=\hat{\ns}_0/2$, using multiplicative Chernoff bound (\cref{lem:chernoff}) with $\alpha=1/2$ , 
\[
\probaOf{m_l>\frac{3}{4}\hat{n}_0}\le \exp\left\{-\frac{\hat{\ns}}{10(\dims+1)}\right\}.
\]
Note that $n_l\sim \binomial{\hat{n}_0}{1-\simprob}$ where $\eta=0.01$ and $\expect{n_l}=\hat{n}_0(1-\eta)$, using the left tail of \cref{lem:chernoff} with $\alpha=\frac{1/4+\eta}{1-\eta}>1/4$,
\[
\probaOf{n_l<\frac{3}{4}\hat{n}_0}\le \exp\left\{-\frac{\hat{\ns}}{32(\dims+1)}\right\}.
\]
Combining the two parts,
\[
\probaOf{m_l>n_l}\le \probaOf{m_l>\frac{3}{4}\hat{n}_0 \text{ or } n_l<\frac{3}{4}\hat{\ns}_0}\le 2\exp\left\{-\frac{\hat{\ns}}{32(\dims+1)}\right\}.
\]
By union bound,
\[
\probaOf{\exists l, m_l>n_l}\le 2(\dims+1)\exp\left\{-\frac{\hat{\ns}}{32(\dims+1)}\right\}.
\]
The probability is at most $\delta'=1/6$ if $\hat{\ns}>32(\dims+1)\ln(12(\dims+1))$. Using $\hat{\ns}=2600\dims(\dims+1)/\eps^2$ samples, we can distinguish between $\p_\rho$ and $\p_{\qkn}$ with probability at least 5/6, which is large enough to ensure $\probaOf{\exists l, m_l>n_l}\le 1/6$. The probability that the entire \cref{alg:no-shared-k} outputs the correct answer is at least $5/6-\delta'=2/3$ as desired.

Recall that $\hat{\ns}=\ns/M$ where $M=\Theta(\dims/\ab)$, therefore the number of copies $\ns=O(\frac{\dims^3}{\ab\eps^2})$ as desired.

\end{proof}

\subsection{Remarks}
We make some important remarks about the algorithms presented in this section.

\paragraph{Random processing of measurement outcomes}  \cref{alg:no-shared-d,alg:no-shared-k} relies on random post-processing of measurement outcomes when generating i.i.d. samples from the uniform distribution over $[\dims+1]$. The goal is to establish a simple reduction to the case where the $\dims+1$ MUBs are sampled uniformly and thus simplify the presentation and proof. However, the formulation of POVM so general that it can be used to describe any classical algorithms. Thus, strictly speaking, the MUB measurements and the random post-processing step can be formulated as a \emph{randomized} measurement scheme.  Nevertheless, it is practically reasonable and relevant to allow randomized post-processing of measurement outcomes as it can be easily implemented with classical computers and provide algorithmic advantages. The lower bound in \cref{thm:fixed-k-lower} is still valid for random classical algorithms applied to the outcomes of fixed measurements, so the copy complexity bound is still tight in this setting. 

To obtain an algorithm using ``truly fixed'' measurements, we can apply any deterministic $\ell_2$ identity testing algorithm using the outcomes with reference distribution $\p_{\qkn}$.

\paragraph{On the algorithm for $\ab<\dims$} \cref{alg:no-shared-k} relies on ad-hoc compression of the outcomes obtained from measuring with the MUBs. When all copies are in one place and one can freely process the measurement outcomes, performing the compression step is not very reasonable. Thus the algorithm is more relevant when the copies are distributed across multiple quantum devices with limited bandwidth to communicate the outcomes to some central server. Nevertheless, after compressing each outcome to $\ell=\log\ab$ bits, the resulting POVM has at most $\ab$ outcomes, and thus the algorithm indeed shows that the lower bound from \cref{thm:fixed-k-lower} is tight.

\bibliography{refs}

\newcommand{\etalchar}[1]{$^{#1}$}
\begin{thebibliography}{BCHJ{\etalchar{+}}21}

\bibitem[Aar20]{Aaronson20}
Scott Aaronson.
\newblock Shadow tomography of quantum states.
\newblock {\em {SIAM} J. Comput.}, 49(5), 2020.

\bibitem[ACFT21]{ACFT:19:IT3}
Jayadev Acharya, Cl{\'{e}}ment~L. Canonne, Cody Freitag, and Himanshu Tyagi.
\newblock Inference under information constraints iii: Local privacy constraints.
\newblock {\em IEEE Journal on Selected Areas in Information Theory}, 2(1):253--267, 2021.

\bibitem[ACL{\etalchar{+}}22]{ACLST22iiuic}
Jayadev Acharya, Clément~L. Canonne, Yuhan Liu, Ziteng Sun, and Himanshu Tyagi.
\newblock Interactive inference under information constraints.
\newblock {\em IEEE Transactions on Information Theory}, 68(1):502--516, 2022.

\bibitem[ACT20a]{acharya2020distributed}
Jayadev Acharya, Cl{\'e}ment~L Canonne, and Himanshu Tyagi.
\newblock Distributed signal detection under communication constraints.
\newblock In Jacob Abernethy and Shivani Agarwal, editors, {\em Proceedings of Thirty Third Conference on Learning Theory}, volume 125 of {\em Proceedings of Machine Learning Research}, pages 41--63. PMLR, 09--12 Jul 2020.

\bibitem[ACT20b]{AcharyaCT19}
Jayadev Acharya, Cl{\'{e}}ment~L. Canonne, and Himanshu Tyagi.
\newblock Inference under information constraints {I:} lower bounds from chi-square contraction.
\newblock {\em {IEEE} Trans. Inf. Theory}, 66(12):7835--7855, 2020.

\bibitem[ACT20c]{ACT:19:IT2}
Jayadev Acharya, Cl\'{e}ment~L. Canonne, and Himanshu Tyagi.
\newblock Inference under information constraints {II}: {C}ommunication constraints and shared randomness.
\newblock {\em IEEE Trans. Inform. Theory}, 66(12):7856--7877, 2020.
\newblock Available at abs/1905.08302.

\bibitem[BCHJ{\etalchar{+}}21]{brandao2021models}
Fernando~GSL Brand{\~a}o, Wissam Chemissany, Nicholas Hunter-Jones, Richard Kueng, and John Preskill.
\newblock Models of quantum complexity growth.
\newblock {\em PRX Quantum}, 2(3):030316, 2021.

\bibitem[BCL20]{BubeckC020}
S{\'{e}}bastien Bubeck, Sitan Chen, and Jerry Li.
\newblock Entanglement is necessary for optimal quantum property testing.
\newblock In Sandy Irani, editor, {\em 61st {IEEE} Annual Symposium on Foundations of Computer Science, {FOCS} 2020, Durham, NC, USA, November 16-19, 2020}, pages 692--703. {IEEE}, 2020.

\bibitem[BFF{\etalchar{+}}01]{BatuFFKRW01}
Tugkan Batu, Lance Fortnow, Eldar Fischer, Ravi Kumar, Ronitt Rubinfeld, and Patrick White.
\newblock Testing random variables for independence and identity.
\newblock In {\em 42nd Annual Symposium on Foundations of Computer Science, {FOCS} 2001, 14-17 October 2001, Las Vegas, Nevada, {USA}}, pages 442--451. {IEEE} Computer Society, 2001.

\bibitem[BHLP20]{brandao2020adversarial}
Fernando G. S.~L. Brandão, Aram~W. Harrow, James~R. Lee, and Yuval Peres.
\newblock Adversarial hypothesis testing and a quantum stein’s lemma for restricted measurements.
\newblock {\em IEEE Transactions on Information Theory}, 66(8):5037--5054, 2020.

\bibitem[BHO20]{barnes2019lower}
Leighton~Pate Barnes, Yanjun Han, and Ayfer Ozgur.
\newblock Lower bounds for learning distributions under communication constraints via fisher information.
\newblock {\em Journal of Machine Learning Research}, 21(236):1--30, 2020.

\bibitem[BKL{\etalchar{+}}19]{BrandaoKLLSW19}
Fernando G. S.~L. Brand{\~{a}}o, Amir Kalev, Tongyang Li, Cedric~Yen{-}Yu Lin, Krysta~M. Svore, and Xiaodi Wu.
\newblock Quantum {SDP} solvers: Large speed-ups, optimality, and applications to quantum learning.
\newblock In Christel Baier, Ioannis Chatzigiannakis, Paola Flocchini, and Stefano Leonardi, editors, {\em 46th International Colloquium on Automata, Languages, and Programming, {ICALP} 2019, July 9-12, 2019, Patras, Greece}, volume 132 of {\em LIPIcs}, pages 27:1--27:14. Schloss Dagstuhl - Leibniz-Zentrum f{\"{u}}r Informatik, 2019.

\bibitem[BO21]{BadescuO21}
Costin Badescu and Ryan O'Donnell.
\newblock Improved quantum data analysis.
\newblock In Samir Khuller and Virginia~Vassilevska Williams, editors, {\em {STOC} '21: 53rd Annual {ACM} {SIGACT} Symposium on Theory of Computing, Virtual Event, Italy, June 21-25, 2021}, pages 1398--1411. {ACM}, 2021.

\bibitem[BOW19]{BadescuO019}
Costin Badescu, Ryan O'Donnell, and John Wright.
\newblock Quantum state certification.
\newblock In Moses Charikar and Edith Cohen, editors, {\em Proceedings of the 51st Annual {ACM} {SIGACT} Symposium on Theory of Computing, {STOC} 2019, Phoenix, AZ, USA, June 23-26, 2019}, pages 503--514. {ACM}, 2019.

\bibitem[Can20]{canonne2020survey}
Cl{\'e}ment~L Canonne.
\newblock A survey on distribution testing: Your data is big. but is it blue?
\newblock {\em Theory of Computing}, pages 1--100, 2020.

\bibitem[Can22]{canonne2022topics}
Clément~L. Canonne.
\newblock Topics and techniques in distribution testing: A biased but representative sample.
\newblock {\em Foundations and Trends® in Communications and Information Theory}, 19(6):1032--1198, 2022.

\bibitem[CCHL21]{ChenCH021}
Sitan Chen, Jordan Cotler, Hsin{-}Yuan Huang, and Jerry Li.
\newblock Exponential separations between learning with and without quantum memory.
\newblock In {\em 62nd {IEEE} Annual Symposium on Foundations of Computer Science, {FOCS} 2021, Denver, CO, USA, February 7-10, 2022}, pages 574--585. {IEEE}, 2021.

\bibitem[CDKS20]{Canonne2020testing}
Clément~L. Canonne, Ilias Diakonikolas, Daniel~M. Kane, and Alistair Stewart.
\newblock Testing bayesian networks.
\newblock {\em IEEE Transactions on Information Theory}, 66(5):3132--3170, 2020.

\bibitem[CHL{\etalchar{+}}23]{chen2023does}
Sitan Chen, Brice Huang, Jerry Li, Allen Liu, and Mark Sellke.
\newblock When does adaptivity help for quantum state learning?
\newblock In {\em 64th {IEEE} Annual Symposium on Foundations of Computer Science, {FOCS} 2023, Santa Cruz, CA, USA, November 6-9, 2023}, pages 391--404. {IEEE}, 2023.

\bibitem[CLHL22]{Chen0HL22}
Sitan Chen, Jerry Li, Brice Huang, and Allen Liu.
\newblock Tight bounds for quantum state certification with incoherent measurements.
\newblock In {\em 63rd {IEEE} Annual Symposium on Foundations of Computer Science, {FOCS} 2022, Denver, CO, USA, October 31 - November 3, 2022}, pages 1205--1213. {IEEE}, 2022.

\bibitem[CLO22]{ChenLO22instance}
Sitan Chen, Jerry Li, and Ryan O'Donnell.
\newblock Toward instance-optimal state certification with incoherent measurements.
\newblock In Po{-}Ling Loh and Maxim Raginsky, editors, {\em Conference on Learning Theory, 2-5 July 2022, London, {UK}}, volume 178 of {\em Proceedings of Machine Learning Research}, pages 2541--2596. {PMLR}, 2022.

\bibitem[Col03]{collins2003moments}
Beno{\^\i}t Collins.
\newblock Moments and cumulants of polynomial random variables on unitarygroups, the itzykson-zuber integral, and free probability.
\newblock {\em International Mathematics Research Notices}, 2003(17):953--982, 2003.

\bibitem[C{\'S}06]{collins2006integration}
Beno{\^\i}t Collins and Piotr {\'S}niady.
\newblock Integration with respect to the haar measure on unitary, orthogonal and symplectic group.
\newblock {\em Communications in Mathematical Physics}, 264(3):773--795, 2006.

\bibitem[DEB{\.Z}10]{durt2010mutually}
Thomas Durt, Berthold-Georg Englert, Ingemar Bengtsson, and Karol {\.Z}yczkowski.
\newblock On mutually unbiased bases.
\newblock {\em International journal of quantum information}, 8(04):535--640, 2010.

\bibitem[DJW13]{duchi2013local}
John~C. Duchi, Michael~I. Jordan, and Martin~J. Wainwright.
\newblock Local privacy and statistical minimax rates.
\newblock In {\em 54th Annual {IEEE} Symposium on Foundations of Computer Science, {FOCS} 2013, 26-29 October, 2013, Berkeley, CA, {USA}}, pages 429--438. {IEEE} Computer Society, 2013.

\bibitem[DK16]{DiakonikolasK16}
Ilias Diakonikolas and Daniel~M. Kane.
\newblock A new approach for testing properties of discrete distributions.
\newblock In Irit Dinur, editor, {\em {IEEE} 57th Annual Symposium on Foundations of Computer Science, {FOCS} 2016, 9-11 October 2016, Hyatt Regency, New Brunswick, New Jersey, {USA}}, pages 685--694. {IEEE} Computer Society, 2016.

\bibitem[DS19]{debrota2019luders}
John~B DeBrota and Blake~C Stacey.
\newblock L{\"u}ders channels and the existence of symmetric-informationally-complete measurements.
\newblock {\em Physical Review A}, 100(6):062327, 2019.

\bibitem[FFGO23]{fawzi:hal-04107265}
Omar Fawzi, Nicolas Flammarion, Aur{\'e}lien Garivier, and Aadil Oufkir.
\newblock {On Adaptivity in Quantum Testing}.
\newblock {\em {Transactions on Machine Learning Research Journal}}, pages 1--33, September 2023.

\bibitem[FL11]{flamia2011direct}
Steven~T. Flammia and Yi-Kai Liu.
\newblock Direct fidelity estimation from few pauli measurements.
\newblock {\em Phys. Rev. Lett.}, 106:230501, Jun 2011.

\bibitem[FO23]{flamian2023tomography}
Steven~T. Flammia and Ryan O'Donnell.
\newblock Quantum chi-squared tomography and mutual information testing.
\newblock {\em CoRR}, abs/2305.18519, 2023.

\bibitem[GKKT20]{guctua2020fast}
Madalin Gu{\c{t}}{\u{a}}, Jonas Kahn, Richard Kueng, and Joel~A Tropp.
\newblock Fast state tomography with optimal error bounds.
\newblock {\em Journal of Physics A: Mathematical and Theoretical}, 53(20):204001, 2020.

\bibitem[Gol17]{goldreich2017introduction}
Oded Goldreich.
\newblock {\em Introduction to property testing}.
\newblock Cambridge University Press, 2017.

\bibitem[HHJ{\etalchar{+}}17]{HaahHJWY17}
Jeongwan Haah, Aram~W. Harrow, Zhengfeng Ji, Xiaodi Wu, and Nengkun Yu.
\newblock Sample-optimal tomography of quantum states.
\newblock {\em {IEEE} Trans. Inf. Theory}, 63(9):5628--5641, 2017.

\bibitem[HKP20]{huang2020predicting}
Hsin-Yuan Huang, Richard Kueng, and John Preskill.
\newblock Predicting many properties of a quantum system from very few measurements.
\newblock {\em Nature Physics}, 16(10):1050--1057, 2020.

\bibitem[KR05]{klappenecker2005mutually}
Andreas Klappenecker and Martin Rotteler.
\newblock Mutually unbiased bases are complex projective 2-designs.
\newblock In {\em Proceedings. International Symposium on Information Theory, 2005. ISIT 2005.}, pages 1740--1744. IEEE, 2005.

\bibitem[KRT17]{kueng2017low}
Richard Kueng, Holger Rauhut, and Ulrich Terstiege.
\newblock Low rank matrix recovery from rank one measurements.
\newblock {\em Applied and Computational Harmonic Analysis}, 42(1):88--116, 2017.

\bibitem[KW21]{khatri2021principles}
Sumeet Khatri and Mark~M Wilde.
\newblock Principles of quantum communication theory: A modern approach.
\newblock 2021.

\bibitem[LA24]{liu2024role}
Yuhan Liu and Jayadev Acharya.
\newblock The role of randomness in quantum state certification with unentangled measurements.
\newblock In Shipra Agrawal and Aaron Roth, editors, {\em Proceedings of Thirty Seventh Conference on Learning Theory}, volume 247 of {\em Proceedings of Machine Learning Research}, pages 3523--3555. PMLR, 30 Jun--03 Jul 2024.

\bibitem[LeC73]{LeCam73}
L.~LeCam.
\newblock Convergence of estimates under dimensionality restrictions.
\newblock {\em The Annals of Statistics}, 1(1):38--53, 1973.

\bibitem[LN22]{lowe2022lower}
Angus Lowe and Ashwin Nayak.
\newblock Lower bounds for learning quantum states with single-copy measurements.
\newblock {\em CoRR}, abs/2207.14438, 2022.

\bibitem[L{\"u}d50]{luders1950zustandsanderung}
Gerhart L{\"u}ders.
\newblock {\"U}ber die zustands{\"a}nderung durch den me{\ss}proze{\ss}.
\newblock {\em Annalen der Physik}, 443(5-8):322--328, 1950.

\bibitem[MdW16]{Montanaro2016}
Ashley Montanaro and Ronald de~Wolf.
\newblock A survey of quantum property testing.
\newblock {\em Theory of Computing}, pages 1--81, 2016.

\bibitem[ON00]{ogawa2000strong}
T.~Ogawa and H.~Nagaoka.
\newblock Strong converse and stein's lemma in quantum hypothesis testing.
\newblock {\em IEEE Transactions on Information Theory}, 46(7):2428--2433, 2000.

\bibitem[OW15]{ODonnellW15}
Ryan O'Donnell and John Wright.
\newblock Quantum spectrum testing.
\newblock In Rocco~A. Servedio and Ronitt Rubinfeld, editors, {\em Proceedings of the Forty-Seventh Annual {ACM} on Symposium on Theory of Computing, {STOC} 2015, Portland, OR, USA, June 14-17, 2015}, pages 529--538. {ACM}, 2015.

\bibitem[OW16]{ODonnellW16}
Ryan O'Donnell and John Wright.
\newblock Efficient quantum tomography.
\newblock In Daniel Wichs and Yishay Mansour, editors, {\em Proceedings of the 48th Annual {ACM} {SIGACT} Symposium on Theory of Computing, {STOC} 2016, Cambridge, MA, USA, June 18-21, 2016}, pages 899--912. {ACM}, 2016.

\bibitem[OW17]{ODonnellW17}
Ryan O'Donnell and John Wright.
\newblock Efficient quantum tomography {II}.
\newblock In Hamed Hatami, Pierre McKenzie, and Valerie King, editors, {\em Proceedings of the 49th Annual {ACM} {SIGACT} Symposium on Theory of Computing, {STOC} 2017, Montreal, QC, Canada, June 19-23, 2017}, pages 962--974. {ACM}, 2017.

\bibitem[Pan08]{Paninski08}
Liam Paninski.
\newblock A coincidence-based test for uniformity given very sparsely sampled discrete data.
\newblock {\em {IEEE} Trans. Inf. Theory}, 54(10):4750--4755, 2008.

\bibitem[Pol03]{Pollard:2003}
David Pollard.
\newblock Asymptopia, 2003.
\newblock Manuscript.

\bibitem[RLW23]{regula2023postselected}
Bartosz Regula, Ludovico Lami, and Mark~M. Wilde.
\newblock Postselected quantum hypothesis testing.
\newblock {\em IEEE Transactions on Information Theory}, pages 1--1, 2023.

\bibitem[Tao23]{tao2023topics}
Terence Tao.
\newblock {\em Topics in random matrix theory}, volume 132.
\newblock American Mathematical Society, 2023.

\bibitem[Ver18]{vershynin2018high}
Roman Vershynin.
\newblock {\em High-dimensional probability: An introduction with applications in data science}, volume~47.
\newblock Cambridge university press, 2018.

\bibitem[Wei78]{weingarten1978asymptotic}
Don Weingarten.
\newblock Asymptotic behavior of group integrals in the limit of infinite rank.
\newblock {\em Journal of Mathematical Physics}, 19(5):999--1001, 1978.

\bibitem[Wri16]{wright2016learn}
John Wright.
\newblock {\em How to learn a quantum state}.
\newblock PhD thesis, Carnegie Mellon University, 2016.

\bibitem[Yu97]{yu1997assouad}
Bin Yu.
\newblock Assouad, fano, and le cam.
\newblock In {\em Festschrift for Lucien Le Cam}, pages 423--435. Springer, 1997.

\bibitem[Yu21]{Yu21sample}
Nengkun Yu.
\newblock Sample efficient identity testing and independence testing of quantum states.
\newblock In James~R. Lee, editor, {\em 12th Innovations in Theoretical Computer Science Conference, {ITCS} 2021, January 6-8, 2021, Virtual Conference}, volume 185 of {\em LIPIcs}, pages 11:1--11:20. Schloss Dagstuhl - Leibniz-Zentrum f{\"{u}}r Informatik, 2021.

\bibitem[Yu23]{Yu2023almost}
Nengkun Yu.
\newblock Almost tight sample complexity analysis of quantum identity testing by pauli measurements.
\newblock {\em IEEE Transactions on Information Theory}, 69(8):5060--5068, 2023.

\end{thebibliography}
\bibliographystyle{alpha}

\appendix
\section{Proof of MIC properties}
\subsection{Proof of~\cref{fact:mic-properties}}
\label{app:fact:mic-properties}
\begin{proof}
    \begin{enumerate}
        \item Since each $\vvec{M_x}\vadj{M_x}$ is p.s.d. and $\Tr[M_x]\ge 0$, $\Choi_{\POVM}$ is also p.s.d. 
        \item \begin{align*}
            \Luders_{\POVM}(\eye_\dims) = \sum_{x}M_x\frac{\Tr[M_x\eye_\dims]}{\Tr[M_x]}=\sum_{x}M_x=\eye_\dims.
        \end{align*}
        The last equality is by the definition of POVM.
        \item For all matrices $X\in\C^{\dims\times\dims}$,
        \begin{align*}
            \Tr[\Luders_{\POVM}(X)]=\sum_{x}\Tr[M_x]\frac{\Tr[M_xX]}{\Tr[M_x]}=\sum_x\Tr[M_xX]=\Tr\left[\sum_{x}M_xX\right]=\Tr[X].
        \end{align*}
        \item Let $X$ be Hermitian, then
        \begin{align*}
            \Luders_{\POVM}(X)^\dagger = \sum_{x}M_x^\dagger\frac{\overline{\Tr[M_xX]}}{\overline{\Tr[M_x]}}=\sum_{x}M_x\frac{\Tr[M_xX]}{\Tr[M_x]}=\Luders_{\POVM}(X)
        \end{align*}
        The second step is because both $X$ and $M_x$ are Hermitian and thus $\Tr[M_xX], \Tr[M_x]$ are real numbers.
    \end{enumerate}
    The proof is complete.
\end{proof}

\subsection{Proof of~\cref{lem:mic-opnorm}}
\label{app:lem:mic-opnorm}
\begin{proof}
   It suffices to prove that for all matrix $X$, $\vadj{X}\Choi_{\POVM}\vvec{X}\le \vvdotprod{X}{X}=\Tr[X^\dagger X]$. Without loss of generality, we can assume that $X$ is Hermitian. Indeed, all matrices $X$ can be written as $X=A+iB$ where both $A$ and $B$ are Hermitian, and so $\vvdotprod{X}{X}=\vvdotprod{A}{A}+\vvdotprod{B}{B}$. If the statement is true for all Hermitian matrices, then
    \begin{align*}
    \vadj{X}\Choi_{\POVM}\vvec{X}&=\vadj{A}\Choi_{\POVM}\vvec{A}+\vadj{B}\Choi_{\POVM}\vvec{B}+i\vadj{A}\Choi_{\POVM}\vvec{B}-i\vadj{B}\Choi_{\POVM}\vvec{A}\\
        &=\vadj{A}\Choi_{\POVM}\vvec{A}+\vadj{B}\Choi_{\POVM}\vvec{B}\\
        &\le \vvdotprod{A}{A}+\vvdotprod{B}{B}=\vvdotprod{X}{X}.
    \end{align*}
    In the second step we used that $\Choi_{\POVM}$ a Hermitian matrix and that $\Luders_{\POVM}$ is Hermitian-preserving, so $\vadj{B}\Choi_{\POVM}\vvec{A}=\vadj{A}\Choi_{\POVM}\vvec{B}\in \R$.

     We now evaluate the expression assuming $X$ is Hermitian.
     \[
     \vadj{X}\Choi_{\POVM}\vvec{X}=\sum_{x}\frac{\vvdotprod{X}{M_x}\vvdotprod{M_x}{X}}{\Tr[M_x]}=\sum_{x}\frac{\Tr[M_xX]^2}{\Tr[M_x]}.
     \]
     Our goal is to prove that $\frac{\Tr[M_xX]^2}{\Tr[M_x]}\le \Tr[M_xX^2]$. If so, then we would have,
     \[
     \vadj{X}\Choi_{\POVM}\vvec{X}=\sum_{x}\frac{\Tr[M_xX]^2}{\Tr[M_x]}\le \sum_x\Tr[M_xX^2]=\Tr\left[\sum_{x}M_xX^2\right]=\Tr[X^2],
     \]
     exactly as desired.
     To prove this statement, let $M_x=\sum_{j}\lambda_{x, j}\qproj{\psi_{x, j}}$ be the eigen-decomposition so that $\sum_{j}\lambda_{x, j}=\Tr[M_x]$. Then,
     \begin{align*}
         \Tr[M_xX]^2&=\sum_{i,j}\lambda_{x, i}\lambda_{x,j}\matdotprod{\psi_{x,i}}{X}{\psi_{x,i}}\matdotprod{\psi_{x,j}}{X}{\psi_{x,j}}\\
         &=\sum_{i,j}\lambda_{x, i}\lambda_{x,j}\Tr[X\qbit{\psi_{x,i}}\qadjoint{\psi_{x,j}}X\qbit{\psi_{x,j}}\qadjoint{\psi_{x,i}}].
     \end{align*}
     By Cauchy-Schwarz,
     \begin{align*}     \Tr[X\qbit{\psi_{x,i}}\qadjoint{\psi_{x,j}}X\qbit{\psi_{x,j}}\qadjoint{\psi_{x,i}}]&=\hdotprod{\qbit{\psi_{x,j}}\qadjoint{\psi_{x,i}}X}{X\qbit{\psi_{x,j}}\qadjoint{\psi_{x,i}}}\\
     &\le \hsnorm{\qbit{\psi_{x,j}}\qadjoint{\psi_{x,i}}X}\hsnorm{X\qbit{\psi_{x,j}}\qadjoint{\psi_{x,i}}}\\
     &=\sqrt{\Tr[X\qbit{\psi_{x,i}}\qdotprod{\psi_{x,j}}{\psi_{x,j}}\qadjoint{\psi_{x,i}}X]\Tr[\qbit{\psi_{x,i}}\qadjoint{\psi_{x,j}}X^2\qbit{\psi_{x,j}}\qadjoint{\psi_{x,i}}]}\\
     &=\sqrt{\matdotprod{\psi_{x,i}}{X^2}{\psi_{x,i}}\matdotprod{\psi_{x,j}}{X^2}{\psi_{x,j}}}\\
     &\le \frac{1}{2}(\matdotprod{\psi_{x,i}}{X^2}{\psi_{x,i}}+\matdotprod{\psi_{x,j}}{X^2}{\psi_{x,j}}).
     \end{align*}
    The final step follows by AM-GM inequality. Plugging in all the expressions,
    \begin{align*}
        \frac{ \Tr[M_xX]^2}{\Tr[M_x]}&\le \frac{\sum_{i,j}\lambda_{x, i}\lambda_{x,j}\frac{1}{2}(\matdotprod{\psi_{x,i}}{X^2}{\psi_{x,i}}+\matdotprod{\psi_{x,j}}{X^2}{\psi_{x,j}})}{\sum_{j}\lambda_{x, j}}\\
        &=\frac{\sum_{j}\lambda_{x, j}\sum_{i}\lambda_{x,i}\matdotprod{\psi_{x,i}}{X^2}{\psi_{x,i}}}{\sum_{j}\lambda_{x, j}}\\
        &=\sum_{i}\lambda_{x,i}\matdotprod{\psi_{x,i}}{X^2}{\psi_{x,i}}\\
        &=\Tr[M_xX^2].
    \end{align*}
    The proof is complete.
\end{proof}

\subsection{Proof of~\cref{lem:mic-trace}}
\label{app:lem:mic-trace}
\begin{proof}
    First, we note that
    \[
    \Tr[\Choi_{\POVM}]=\sum_{x}\frac{\vvdotprod{M_x}{M_x}}{\Tr[M_x]}=\sum_x\frac{\hsnorm{M_x}^2}{\Tr[M_x]}.
    \]
    Since $\hsnorm{M_x}\le\tracenorm{M_x}= \Tr[M_x]$, we have
    \[
    \Tr[\Choi_{\POVM}]\le \sum_x\frac{{\Tr[M_x]}^2}{\Tr[M_x]}=\sum_{x}\Tr[M_x]=\dims.
    \]
    We also note that since $M_x\preceq \eye_\dims$,  $\opnorm{M_x}\le 1$. By H\"older's inequality, 
    $\hsnorm{M_x}^2\le \opnorm{M_x}\Tr[M_x]$. Thus, when the size of $\POVM$ is at most $\ab$,
    \[
    \Tr[\Choi_{\POVM}]\le\sum_{x=1}^{\ab}\opnorm{M_x}\le \ab.\qedhere
    \]
    Since $\Choi_{\POVM}$ is p.s.d by~\cref{fact:mic-properties}, we have $\tracenorm{\Choi_{\POVM}}=\Tr[\Choi_{\POVM}]$, thereby completing the proof.
\end{proof}

\section{Missing proofs in the lower bound}
\label{app:lower-bound-proof}

\subsection{Proof of~\cref{thm:rand-mat-opnorm-concentration}}
\label{app:prop:perturbation-trace-distance}

\begin{proof}
    We first prove that for any fixed unit vector $x\in \C^{\dims}$, the norm of $Wx$ is at most $O(\sqrt{\dims})$ with high probability. Then we use an $\epsilon$-net argument to show that the probability is also high for \textit{all} unit vectors. We start with the following lemma.
    \begin{lemma}
    \label{lem:fixed-vec-opnorm-concentration}
        Let $\{\ptb_i\}_{i=1}^{\dims^2}, \{V_i\}_{i=1}^{\dims^2}$ and $W$ be defined in~\cref{thm:rand-mat-opnorm-concentration}. Then there exists a universal constant $c'$ for any fixed unit vector $x$ and all $s>0$,
        \[
        \probaOf{\normtwo{Wx}\ge (1+s)\sqrt{\dims}}\le 2\exp\{-c's^2\dims\}.
        \]
    \end{lemma}
    \begin{proof}
        Let $\ptb = (\ptb_1, \ldots, \ptb_{\dims^2})\in \R^{\dims^2}$, and $\Pi_\ell\in \R^{\dims^2\times\dims^2}$ be a diagonal matrix with 1 in the first $\ell$ diagonal entries and 0 everywhere else. Then 
        \[
        Wx=\sum_{i=1}^{\ell}\ptb_iV_ix = V_x\Pi_\ell \ptb,
        \]
        where 
        \begin{equation*}
            V_x\eqdef[V_1x, \ldots, V_{\dims^2}x]\in\C^{\dims\times \dims^2}
        \end{equation*}
        which is an isometry, i.e. $V_xV_x^\dagger = \eye_\dims$, as stated in~\cref{claim:isometry} which will be proved at the end of this section. Therefore,
        \[
        \opnorm{V_x}=1, \quad \hsnorm{V_x}^2=\Tr[V_xV_x^\dagger]=\dims.
        \]From this, we can apply concentration for linear transforms of independent sub-Gaussian random variables.
        \begin{theorem}[{\citet[Theorem 6.3.2]{vershynin2018high}}]
            Let $B\in \C^{m\times n}$ be a fixed $m\times n$ matrix and let $X=(X_1, \ldots, X_n)\in\R^n$ be a random vector with independent, mean zero, unit variance, and sub-Gaussian coordinates with Orlicz-2 norm $\|X_i\|_{\psi_2}\le K$. Then there exists a universal constant $C=\frac{3}{8}$ such that for all $t>0$,
            \[
            \probaOf{|\normtwo{BX}-\hsnorm{B}|>t}\le 2\exp\left\{-\frac{Ct^2}{K^4\opnorm{B}^2}\right\}.
            \]
        \end{theorem}
        \begin{remark}
            The original~\cite[Theorem 6.3.2]{vershynin2018high} was stated for real matrix $B$. However, it is straightforward to extend the argument to complex $B$ by considering $\tilde{B}=\begin{bmatrix}
                \Real(B)\\
                \Img(B)
            \end{bmatrix}$. Then $\opnorm{\tilde{B}}=\opnorm{B}, \hsnorm{\tilde{B}}=\hsnorm{B}$, and $\normtwo{BX}=\normtwo{\tilde{B}X}$. 
        \end{remark}
        Setting $B=V_x\Pi_\ell$, we observe that
        \[
        \opnorm{B}\le \opnorm{V_x}\opnorm{\Pi_\ell}=1, \quad  \hsnorm{B}\le \hsnorm{V_x}=\sqrt{\dims}.
        \]
        Thus, plugging $t=s\sqrt{\dims}$, and noting that $\|\ptb_i\|_{\psi_2}=1/\sqrt{\ln 2}=K$, we have
        \[
        \probaOf{\normtwo{Wx}> (1+s)\sqrt{\dims}}\le \probaOf{\normtwo{B\ptb}>s\sqrt{\dims}+\hsnorm{B}}\le 2\exp\left\{-C\dims(\ln2)^2s^2\right\}.
        \]   
        Setting $c'=C(\ln2)^2=\frac{3(\ln2)^2}{8}$ completes the proof.
    \end{proof}
    We can then proceed to use the $\epsilon$-net argument, which follows closely to \cite[Section 2.3]{tao2023topics}.

    \begin{lemma}[{\cite[Lemma 2.3.2]{tao2023topics}}]
        Let $\Sigma$ be a maximal $1/2$-net of the unitary sphere, i.e., a maximal set of points that are separated from each other by at least $1/2$. Then for any matrix $M\in\C^{\dims\times\dims}$ and $\lambda>0$,
        \[
         \probaOf{\opnorm{M}>\lambda}\le \sum_{y\in\Sigma}\probaOf{\normtwo{My}>\lambda/2}.
        \]
    \end{lemma}
    By standard volume packing argument, the size of $\Sigma$ is at most $\exp(O(\dims))$, 
    \begin{lemma}[{\cite[Lemma 2.3.4]{tao2023topics}}]
    \label{lem:packing-argument}
        Let $\epsilon\in(0, 1)$ and let $\Sigma$ be an $\epsilon$-net of the unit sphere. Then $|\Sigma|\le (C'/\epsilon)^\dims$ where $C'=3$.
    \end{lemma}
    Thus with $c'$ defined in~\cref{lem:fixed-vec-opnorm-concentration} and $C'$ defined in~\cref{lem:packing-argument} we conclude that
    \[
    \probaOf{\opnorm{W}>2(1+s)\sqrt{\dims}}\le 2(2C')^\dims \exp\{-c's^2\dims\}=2\exp\left\{-(c's^2-\ln(2C'))\dims\right\}.
    \]
    Thus choosing $s$ sufficiently large, we can guarantee that the tail probability decays exponentially in $\dims$. Specifically, let $\alpha>0$ and $s^2=\frac{\alpha+\ln(2C')}{c'}$, then we have
    \[
    \probaOf{\opnorm{W}>2(1+s)\sqrt{\dims}}\le 2e^{-\alpha\dims}.
    \]
    Setting $\cop_\alpha = 2(1+s)=2\Paren{1+\sqrt{\frac{\alpha+\ln(2C')}{c'}}}$ proves the theorem. In particular, $\kappa_1\le 10$ when substituting the values of $c'$ and $C'$.
\end{proof}

We end this section with the proof of the isometry claim.
\begin{claim}
 \label{claim:isometry}
    Let $V_1, \ldots, V_{\dims^2}$ be an orthonormal basis of $\C^{\dims\times \dims}$ and $x\in\C^\dims$ be a unit vector. Then $V_x\eqdef[V_1x, \ldots, V_{\dims^2}x]\in\C^{\dims\times \dims^2}$ is an isometry: $V_xV_x^\dagger = \eye_\dims$.
\end{claim}
\begin{proof}
 Let $V_x\supparen{k}$ be the $k$th row of $V_x$ written as row vector. It suffices to prove that
 \[
 V_x\supparen{k}(V_x\supparen{l})^{\dagger}=\delta_{kl}
 \]
 Let $V_i^{(k)}$ be the $k$th row of $V_i$, written as a row vector. Then the $k$th element of $V_ix$ is
        \begin{equation*}
            \quad v_i\supparen{k}\eqdef V_i^{(k)}x.
        \end{equation*}
     Since $V_1, \ldots, V_{\dims^2}$ are orthonormal, we know that
        \[
        V\eqdef[\VecOp(V_1), \ldots, \VecOp(V_{\dims^2})]
        \]
        is a unitary matrix in $\C^{\dims^2\times \dims^2}$. Let $V^{j}$ be the $j$th row of $V$, then because $V$ is unitary, the vector dot product $\hdotprod{V^{j}}{ V^{i}}=\delta_{ij}$. Let 
        \[
        V^{(k)}=[(V^{k})^{\dagger}, (V^{k+\dims})^{\dagger}, \ldots(V^{k+\dims(j-1)})^{\dagger}, \ldots, (V^{k+\dims(d-1)})^{\dagger}]^\dagger
        \]
        which picks out the $k$th row of all $V_1, \ldots, V_{\dims^2}$. Then, we have
        \[
        V^{(k)}=[(V_1^{(k)})^\top, \ldots, (V_{\dims^2}^{(k)})^\top]. 
        \]
        Thus,
        \[
        \sum_{i=1}^{\dims^2}(V_i^{(k)})^{\dagger}V_i^{(k)}=\overline{V^{(k)}(V^{(k)})^{\dagger}}=\eye_\dims,
        \]
        and for $k\ne l$,
        \[
        \sum_{i=1}^{\dims^2}(V_i^{(k)})^{\dagger}V_i^{(l)}=\overline{V^{(k)}(V^{(l)})^{\dagger}}=0.
        \]
        Therefore, 
        \[
         V_x\supparen{k}(V_x\supparen{l})^{\dagger}=\sum_{i=1}^{\dims^2}v_i\supparen{k}(v_i\supparen{l})^\dagger=\sum_{i=1}^{\dims^2}  x^{\dagger}(V_i^{(l)})^{\dagger}V_i^{(k)}x=x^\dagger\delta_{kl}\eye_\dims x=\delta_{kl},
        \]
        exactly as desired, completing the proof.
\end{proof}

\subsection{Proof of~\cref{thm:chi-square-upper-bound}}
\label{app:thm:chi-square-upper-bound}
We first prove a general upper bound for~\cref{lem:chi-square-expansion} with $\sigma,\sigma'\sim\ptbDistr(\hbasis)$.
\begin{restatable}{theorem}{chisqub}
\label{thm:chi-square-upper-bound-intermediate}
    Let $\ell\in[\frac{\dims^2}{2},\dims^2-1]$, $\hbasis=(V_1, \ldots, V_{\dims^2}=\frac{\eye_\dims}{\sqrt{\dims}})$ be an orthonormal basis of $\Herm{\dims}$, and $ V\eqdef[\vvec{V_1}, \ldots, \vvec{V_\ell}]$,  $\sigma, \sigma'\sim \ptbDistr(\hbasis)$ defined in~\cref{def:perturbation}. Then for $\ns<\frac{\dims^2}{6\cd^2\eps^2\opnorm{V^\dagger \avgChoi V}}$,
    \begin{equation}
        \expectDistrOf{\sigma,\sigma'}{\exp\left\{\ns\dims \vadj{\barDelta_{\sigma'}}\avgChoi\vvec{\barDelta_\sigma}\right\}}\le \exp\left\{\frac{\cd^2\ns^2\eps^4}{\ell^2}\hsnorm{V^\dagger \avgChoi V}^2\right\}+\frac{4}{e^\dims}.
        \label{equ:chi-square-final-ub}
    \end{equation} 
\end{restatable}
\begin{proof}
    We parameterize the distribution $\ptbDistr(\hbasis)$ using $\ptb,\ptb'\sim\{-1,1\}^{\ell}$ and $\Delta_z,\barDelta_z,\sigma_z$ defined in~\cref{def:perturbation}.
    
    First, we claim that due to the exponentially small probability of the bad event $\sigma_z=\Delta_z+\qmm\notin \mathcal{P}_\eps$ as stated in~\cref{prop:perturbation-trace-distance}, we can consider $\Delta_z$ instead of the normalized perturbation $\barDelta_z$. The claim is proved at the end of this section.
    \begin{lemma} 
    \label{claim:projected-chi-square-ub}
    Let $z,z'\sim\{-1,1\}^\ell$ be uniform and $\barDelta_z$ and $\Delta_z$ be defined in~\cref{def:perturbation}, then
        $$\expectDistrOf{\ptb, \ptb'}{\exp\left\{\ns\dims\vadj{\barDelta_{\ptb}}\avgChoi\vvec{\barDelta_\ptb}\right\}}\le \expectDistrOf{\ptb, \ptb'}{\exp\left\{\ns\dims\vadj{\Delta_{\ptb'}}\avgChoi\vvec{\Delta_\ptb}\right\}}+\frac{4}{e^\dims}.$$       
    \end{lemma}

    We then apply a standard result on the moment generating function of Radamacher chaos.
    \begin{lemma}[{\citet[Claim IV.17]{AcharyaCT19}}]
    \label{lem:radamacher-mgf}
        Let $\ptb, \ptb'$ be two independent random vectors distributed uniformly over $\{-1, 1\}^\ell$. Then for any positive semi-definite real matrix $H$, 
        \[
        \log\expectDistrOf{\ptb,\ptb'}{\exp\{\lambda \theta^\top H\theta'\}}\le\frac{\lambda^2}{2}\frac{\|H\|_{HS}^2}{1-4\lambda^2\opnorm{H}^2}, \quad\text{for }0\le \lambda<\frac{1}{\opnorm{H}}.
        \]
    \end{lemma}
    We now evaluate the inner product. Note that $\Choi_i$ and $\avgChoi$ are p.s.d. Hermitian matrices.   
    Setting $V=[\VecOp(V_1), \ldots, \VecOp(V_\ell)]\in \C^{\dims^2\times \ell}$, we have $\vvec{\Delta_z}=\frac{\cd\eps}{\sqrt{\dims\ell}}Vz$. Therefore, set $H\eqdef V^\dagger\avgChoi V$,
    \begin{align}
    \label{equ:mic-inner-product}
        \vadj{\barDelta_{\ptb'}}\avgChoi\vvec{\barDelta_\ptb}=\frac{\cd^2\eps^2}{\dims\ell}\ptb^\dagger V^\dagger\avgChoi V\ptb' = \frac{\cd^2\eps^2}{\dims\ell}z^\dagger Hz'.
    \end{align}
    We now show that $H$ is a real matrix when each $V_i$ is Hermitian. First note that $\avgChoi\vvec{V_j} =\vvec{\avgLuders(V_j)}$. Therefore element $i,j$ in $H$ is
    \begin{equation}
    \label{equ:H-element}
        H_{ij}=\vadj{V_i}\avgChoi\vvec{V_j})=\hdotprod{V_i}{\avgLuders(V_j)}\in \R.
    \end{equation}

    We used the fact that $\avgLuders$ is Hermiticity preserving and thus $\avgLuders(V_j)$ is Hermitian. Since $\Herm{\dims}$ is a real Hilbert space, the inner product is a real number.

    We then set $\lambda=\frac{\cd^2\ns\eps^2}{\ell}$ in~\cref{lem:radamacher-mgf}. Thus for $\ns<\frac{\ell}{3\cd^2\eps^2\opnorm{H}}$, we have
    
    $$\lambda\opnorm{H}\le \frac{\cd^2\eps^2}{\ell}\cdot \frac{\ell}{3\cd^2\eps^2\opnorm{H}}\cdot\opnorm{H}= \frac{1}{3}. $$
    The condition on $\lambda$ in~\cref{lem:radamacher-mgf} is satisfied. Hence, combining~\eqref{equ:mic-inner-product} and~\cref{lem:radamacher-mgf},
    \begin{align*}
        \expectDistrOf{\ptb, \ptb'}{\exp\left\{\ns\dims\vadj{\Delta_{\ptb'}}\avgChoi\vvec{\Delta_\ptb}\right\}}&=\expectDistrOf{z, z'}{\exp\left\{\frac{\cd^2\ns\eps^2}{\ell}z^\top Hz'\right\}}\\
        &\le\exp\left\{\frac{\lambda^2}{2(1-4\lambda^2\opnorm{H}^2)}\hsnorm{H}^2\right\}\\
        &< \exp\left\{\frac{\cd^4\ns^2\eps^4}{\ell^2}\hsnorm{H}^2\right\}.
    \end{align*}
    The final inequality  is because $\lambda\opnorm{H}\le 1/3$, and thus $2(1-4\lambda^2\opnorm{H}^2)\ge 2(1-4/9)>1$.
    Combining with~\cref{claim:projected-chi-square-ub} and plugging in $H=V^\dagger\avgChoi V$ proves~\cref{thm:chi-square-upper-bound}.     
\end{proof}

Set $H=V^\dagger\avgChoi V$. To finish the proof of~\cref{thm:chi-square-upper-bound}, we obtain different bounds for $\hsnorm{V^\dagger \avgChoi V}$ for different choices of the basis $\hbasis$. For all orthonormal basis $\hbasis=(V_1, \ldots, V_{\dims^2})$ of $\Herm{\dims}$, we must have $V^\dagger V=\eye_{\ell}$, and thus $V$ is an isometry and $\opnorm{V}=\opnorm{V^\dagger}= 1$. Using $\hsnorm{AB}\le \opnorm{A}\hsnorm{B}$, we obtain a sequence of inequalities for the norms,
    \begin{align}
        \opnorm{V^\dagger\avgChoi V}&\le \hsnorm{V^\dagger\avgChoi V}\le \hsnorm{\avgChoi}\label{equ:norm-relation-general-1}\\
        &\le \frac{1}{\ns}\sum_{i=1}^{\ns}\hsnorm{\Choi_i}\le \max_{\POVM\in\povmset}\hsnorm{\Luders_{\POVM}}
        \label{equ:norm-relation-general-2}
    \end{align}
The second line is due to triangle inequality. Therefore, 
\[
\ns\le \frac{\dims^2}{6\cd^2\eps^2\max_{\POVM\in\povmset}\hsnorm{\Choi_{\POVM}}}\implies \ns\le \frac{\dims^2}{6\cd^2\eps^2\opnorm{V^\dagger\avgChoi V}}.
\]
Thus the condition on $\ns$ in~\cref{thm:chi-square-upper-bound-intermediate} is satisfied. Applying the theorem and~\eqref{equ:norm-relation-general-1}~\eqref{equ:norm-relation-general-2},
\[
\expectDistrOf{\sigma,\sigma'}{\exp\left\{\ns\dims \vadj{\barDelta_{\sigma'}}\avgChoi\vvec{\barDelta_\sigma}\right\}}\le \exp\left\{\frac{\cd^2\ns^2\eps^4}{\ell^2}\max_{\POVM\in\povmset}\hsnorm{\Luders_{\POVM}}^2\right\}+\frac{4}{e^\dims}.
\]
Substituting $\ell=\dims^2/2$ proves the first part~\cref{equ:chi-square-random} of~\cref{thm:chi-square-upper-bound}.

When $\hbasis=(V_1, \ldots, V_{\dims^2})$ is the eigenbasis of $\avgLuders$, let $\lambda_i$ be the eigenvalue of $V_i$, and by assumption $\lambda_1\le\cdots\le\lambda_{\dims^2}$. In this case, using~\eqref{equ:H-element}
    \begin{align*}
        H_{ij}=\hdotprod{V_i}{\avgLuders(V_j)}=\lambda_j\hdotprod{V_i}{V_j}=\lambda_j\delta_{ij}.
    \end{align*}
    Therefore $H=V^\dagger\avgChoi V=\diag\{\lambda_1, \ldots, \lambda_\ell\}$, and  $\hsnorm{H}^2=\sum_{i=1}^\ell \lambda_i^2$. This can be bounded in terms of the trace norm $\tracenorm{\avgLuders}$. Recall the eigenvalues are sorted in increasing order, 
    \[
    \lambda_\ell\le \frac{1}{\dims^2-\ell}\sum_{i=\ell+1}^{\dims^2}\lambda_i\le\frac{\tracenorm{\avgLuders}}{\dims^2-\ell}.
    \]
    The final inequality is due to $\tracenorm{\avgLuders}=\sum_{i}\lambda_i$. Therefore,
    \begin{align}
    \hsnorm{H}^2=\sum_{i=1}^{\ell}\lambda_i^2\le \ell\lambda_\ell^2\le \ell\Paren{\frac{\hsnorm{\avgLuders}}{\dims^2-\ell}}^2=\frac{2\tracenorm{\avgLuders}^2}{\dims^2}
    \label{equ:fixed-hsnorm-bound}
    \end{align}
    By linearity of trace,  $\tracenorm{\avgLuders}=\frac{1}{\ns}\sum_{i=1}^{\ns}\tracenorm{\Luders_i} \le \max_{\POVM\in\povmset}\tracenorm{\Luders_{\POVM}}$. 
    Using~\eqref{equ:norm-relation-general-1} and~\eqref{equ:fixed-hsnorm-bound},
    \[
    \frac{\dims^3}{6\sqrt{2}\cd^2\eps^2\max_{\POVM\in\povmset}\tracenorm{\Luders_{\POVM}}}\le \frac{\dims^3}{6\sqrt{2}\cd^2\eps^2{\dims\hsnorm{H}}/{\sqrt{2}}}\le \frac{\dims^3}{6\cd^2\eps^2\dims\opnorm{H}}.
    \]
    The condition is satisfied as long as $\ns$ is upper bounded by the first expression in the above equation. Therefore, by~\cref{thm:chi-square-upper-bound-intermediate},
    \[
    \expectDistrOf{\sigma,\sigma'}{\exp\left\{\ns\dims \vadj{\barDelta_{\sigma'}}\avgChoi\vvec{\barDelta_\sigma}\right\}}\le \exp\left\{\frac{2\cd^2\ns^2\eps^4}{\ell^2\dims^4}\max_{\POVM\in\povmset}\tracenorm{\Luders_{\POVM}}^2\right\}+\frac{4}{e^\dims}.
    \]
    When $\dims\ge 16$, the extra term $4/\exp\{\dims\}$ is negligible. Thus we complete the proof of~\cref{thm:chi-square-upper-bound}.

\subsubsection{Proof of~\cref{claim:projected-chi-square-ub}}
\begin{proof}   
        Note that $\barDelta_z=a_z\Delta_z$ where
    \[
    a_z\eqdef\min\left\{1, \frac{1}{\dims\opnorm{\Delta_z}}\right\}\in[0, 1].
    \]
    Therefore 
    \[
    \vadj{\barDelta_{\ptb}}\avgChoi\vvec{\barDelta_\ptb}=a_za_{z'}\vadj{\Delta_{\ptb}}\avgChoi\vvec{\Delta_\ptb}.
    \]
        As a short hand let $f(z, z')=\ns\dims\vadj{\barDelta_{\ptb}}\avgChoi\vvec{\barDelta_\ptb}$. Denote event $E$ as $f(z,z')<0$ and $a_za_{z'}< 1$. When this event occors, $\exp\{a_za_{z'}f(z,z')\}\le 1$. Using~\cref{prop:perturbation-trace-distance}, let $\delta= 2\exp(-\dims)$,
        \[
        \probaOf{a_{z}<1}\le \delta.
        \]
        Thus by union bound,
        \[
        \probaOf{E}\le \probaOf{a_{z}a_{z'}<1}=\probaOf{a_{z}<1 \text{ or } a_{z'}<1}\le 2\delta.
        \]
        Note that $E^c$ denotes the event that $f(z,z')\ge 0$ or $a_za_z'=1$. When this occurs, $a_za_{z'}f(z,z')\le f(z,z')$. Thus,
        \begin{align*}
            &\quad \expectDistrOf{z, z'}{\exp\left\{a_za_{z'}f(z,z')\right\}}\\
            &=\expectCondDistrOf{z,z'}{\exp\left\{ a_za_{z'}f(z,z')\right\}}{E^c}\probaOf{E^c} +\expectCondDistrOf{z,z'}{\exp\left\{ a_za_{z'}f(z,z')\right\}}{E}\probaOf{E}\\
            &\le \expectCondDistrOf{z,z'}{\exp\left\{f(z,z')\right\}}{E^c}\probaOf{E^c}+2\delta\\
            &\le  \expectDistrOf{z,z'}{\exp\left\{f(z,z')\right\}}+2\delta,
        \end{align*}
        as desired. The second-to-last inequality uses $a_zf_z'f(z, z')\le 0$ when event $E$ happens, and the final inequality uses $\exp\{f(z,z')\}>0$ and therefore
        \begin{align*}
            \expectDistrOf{z,z'}{\exp\left\{f(z,z')\right\}}&=\expectCondDistrOf{z,z'}{\exp\left\{f(z,z')\right\}}{E^c}\probaOf{E^c}+\expectCondDistrOf{z,z'}{\exp\left\{f(z,z')\right\}}{E}\probaOf{E}\\
            &\ge \expectCondDistrOf{z,z'}{\exp\left\{f(z,z')\right\}}{E^c}\probaOf{E^c}.
        \end{align*}
        Plugging in the definition of $a_z$ and $f(z,z')$ completes the proof.
    \end{proof}

\section{Proof of quantum domain compression}
\label{app:lem:l2-norm-haar}
\subsection{Weingarten calculus}
\subsubsection{Partitions and permutations}
Let $d, q$ be positive integers. A partition $\lambda\vdash q$ is an integer vector $\lambda=(\lambda_1, \ldots, \lambda_q)$ where $\lambda_1\ge \cdots \ge \lambda_q\ge 0$ and $\sum_{i}\lambda_i=q$. $\lambda$ can also be represented in terms of the unique elements and their multiplicity. For example, $(2^1 1^2)$ represents the partition $\lambda=(2, 1, 1)$. 

A permutation $\pi:[n]\mapsto[n]$ is a bijection over $[n]$. Let $\Sim_n$ be the group of permutations over $[n]$ (also called the symmetric group). Each permutation $\pi\in\Sim_n$ can be decomposed into cycles. 
We denote $\cycle(\pi)$ as the set of cycles in $\pi$. For example, for the permutation $\pi = (2,3,1,4)$, the cycles are
\[
\mathcal{C}(\pi)=\{(1,2,3), (4)\}.
\]

The cycle lengths $|c|, c\in \cycle(\pi)$ form a partition of $n$, which is $(31)$ for the previous example. We say that $\pi$ has a cycle type $\lambda\vdash n$ if the cycle lengths form a partition $\lambda$. We can thus group the permutations according to their cycle types $\lambda$. Let $\{\lambda\}$ be the set of permutations with cycle type $\lambda$. We sometimes abuse notation and use $\lambda$ to denote $\{\lambda\}$ when it is clear from the context.

We use $e$ to denote the identity permutation, which has cycle type $(1^n)$.

\subsubsection{Weingartun functions}

The Weingarten function $\Wg_{d, q}:\Sim_q\mapsto \R$ is defined for permutations in $\Sim_q$. If $\pi$ and $\tau$ have the same cycle type, then $\Wg_{d, q}(\pi)=\Wg_{d, q}(\tau)$. Thus we can parameterize the input to $\Wg_{d, q}$ by partitions of $q$.  We drop $q$ in the subscript when it can be clearly defined by the permutation $\pi$. 

The precise definition can be found in~\cite{collins2003moments,collins2006integration} It involves Schur polynomials and character functions which we do not need in our proof. Below are some simple examples of Weingarten functions with $q=1,2$.
\begin{equation}
    \quad\Wg_{d}(1)= \frac{1}{d}, \quad\Wg_{d}(2)=-\frac{1}{d(d^2-1)},\quad \Wg_{d}(1^2)=\frac{1}{d^2-1}.
    \label{equ:wg-deg2}
\end{equation}

The asymptotic behavior of $\Wg_{d}$ is characterized by the lemma below,
\begin{lemma}[{\cite[Corollary 2.7]{collins2006integration}}]
Let $\cycle(\pi)$ be the set of cycles in $\pi\in\Sim_{q}$, and $|\pi|\eqdef q-|\cycle(\pi)|$.
\label{lem:wg-asymptotic}
    As $\dims\to\infty$,
    \[
    \Wg_{\dims}(\pi) = (1+O(\dims^{-2}))\Mob(\pi)\dims^{-q-|\pi|}.
    \]
Here $\Mob(\pi)=\prod_{c\in \cycle(\pi)}(-1)^{|c|-1}\Cat_{|c|-1}$,  $|c|$ is the length of cycle $c$, and $\Cat_n=\frac{1}{2n+1}{2n+1\choose n}$ is the Catalan number.
\end{lemma}

Thus, $\Wg_{\dims,q}$ is a rational function with respect to $\dims$ with degree at most $\dims^{-q}$, which is achieved if and only if $\pi=e$, the identity. The leading constant factor is determined by $\Mob(\pi)$. We compute the asymptotics of $\Wg_\dims(\pi)$ for $\pi\in\Sim_4$ in~\cref{tab:wg-deg-4} which we will use in our proof.
\begin{table}[]
    \centering
    \def\arraystretch{1.4}
    \begin{tabular}{|c|c|c|c|c|c|}
    \hline
         $\lambda$& $(1^4)$ & $(21^2)$ &$(2^2)$ & $(31)$ &$(4)$ \\ \hline
         $|\{\lambda\}|$& 1& 6&3 &8 &6\\\hline
         $\Wg_\dims(\pi)$& $\dims^{-4}$ & $-{\dims^{-5}}$ & $\dims^{-6}$  &$2\dims^{-6} $ & $-5\dims^{-7}$ \\\hline
    \end{tabular}
    \caption{Size of each equivalent class $\{\lambda\}$ and the asymptotics of $\Wg_{\dims}$ for $\Sim_4$.}
    \label{tab:wg-deg-4}
\end{table}

\subsubsection{Computing Haar integrals}
For each permutation $\pi\in\Sim_{\ell}$, we can define the permutation operator $P_\pi\in \C^{\dims^\ell\times\dims^\ell}$, which acts on $\C^{\dims^\ell}$ as
\[
P_\pi\qbit{x_1}\otimes\cdots\otimes\qbit{x_\ell}=\qbit{x_{\pi(1)}}\otimes\cdots\otimes\qbit{x_{\pi(\ell)}}.
\]
Let $\Haar{\dims}$ be the Haar measure over unitary matrices in $\C^{\dims\times\dims}$. ~\cref{lem:haar-twirl} helps us to compute expectations of unitary transformations of matrices.
\begin{lemma}[\cite{brandao2021models}, Eq.7.32]
\label{lem:haar-twirl}
    For all matrix $M\in \C^{\dims^\ell\times\dims^\ell}$,
    \[
    \expectDistrOf{U\sim \Haar{\dims}}{(U^\dagger)^{\otimes\ell}MU^{\otimes\ell}}=\sum_{\pi,\tau\in\Sim_{\ell}}\Wg_\dims(\pi^{-1}\tau)P_\pi\Tr[P_\tau M].
    \]
\end{lemma}
For $A\in\C^{\dims\times\dims}$ and a permutation $\pi\in \Sim_{\ell}$, we define
\begin{equation}
    \permProd{A}{\pi}\eqdef\Tr[P_\pi A^{\otimes \ell}]=\prod_{c\in\cycle(\pi)}\Tr[A^{|c|}].
\label{equ:perm-prod}
\end{equation}
The following useful lemma is a corollary of~\cref{lem:haar-twirl}.
\begin{lemma}[\cite{ChenLO22instance}, Lemma 3.13]
\label{lem:haar-trace}
    For $\dims\ge 2$, $\ell$ positive integer, $A,B\in\C^{\dims\times\dims}$, we have
    \[
    \expectDistrOf{U\sim \Haar{\dims}}{\Tr[AU^\dagger BU]^\ell}=\sum_{\pi, \tau\in \Sim_\ell}\Wg_\dims(\pi^{-1}\tau)\permProd{A}{\pi}\permProd{B}{\tau}.
    \]
    When $\ell=1$, $\expectDistrOf{U\sim \Haar{\dims}}{\Tr[AU^\dagger BU]}=\Tr[A]\Tr[B]/\dims$.
\end{lemma}

\subsection{Proof of~\cref{lem:l2-norm-haar}}
\label{app:proof:12-norm-haar}
We first recall some definitions. Let $\ab\le \dims$ be powers of 2 and $\PiRank\eqdef \dims/\ab$. Let $U=[\qbit{u_1}, \ldots, \qbit{u_\dims}]$ be drawn from the Haar measure. We define the projections $\Pi_x^{U}\eqdef\sum_{i=(\PiRank-1)x+1}^{\PiRank x}\qproj{u_i}$, which divides $\C^\dims$ into orthogonal subspaces with equal dimensions. Let $\POVM_U=\{\Pi_x^U\}_{x=1}^\ab$. The distribution $\p_\rho^{U}$ is defined by the Born's rule where $\p_\rho^{U}(x)=\Tr[\Pi_x^U\rho]$. 

For brevity, we omit all the superscripts and subscripts of $U$ when it is clear from the context. The projection matrices satisfy $(\Pi_x)^2=\Pi_x$ and $\Tr[\Pi_x]=\PiRank=\dims/\ab$. Furthermore, we can write 
$\Pi_x=U D_x U^\dagger,$
where $D_x$ is a diagonal matrix with 1 at the $(\PiRank-1)x+1$ to $\PiRank x$ diagonal entries and 0 elsewhere.

We now compute the $\ell_2$ norms. Let $\Delta=\rho-\qkn$,
\[
\normtwo{\p_\rho}^2=\sum_{x=1}^{\ab}\Tr[\Pi_x\rho]^2,\quad \normtwo{\p_\rho-\p_{\qkn}}^2=\sum_{x=1}^{\ab}\Tr[\Pi_x\Delta]^2.
\]
By symmetry of the Haar measure, 
\[
\expectDistrOf{U}{\normtwo{\p_\rho}^2}=\ab\expectDistrOf{U}{ \Tr[\Pi_1\rho]^2}, \quad\expectDistrOf{U}{\normtwo{\p_\rho-\p_{\qkn}}^2}=\ab \expectDistrOf{U}{\Tr[\Pi_x\Delta]^2}.
\]
Both quantities can be computed using the lemma below,
\begin{lemma}
\label{lem:haar-square}
    Let $M\in \C^{\dims\times\dims}$ be a matrix, $U=[\qbit{u_1}, \ldots, \qbit{u_\dims}]$ drawn from the Haar measure, and $\Pi_1=\sum_{i=1}^{\PiRank}\qproj{u_i}$ where $\PiRank=\dims/\ab$. Then
    \[
    \expectDistrOf{U}{\Tr[\Pi_1 M]^2}=\frac{1}{\ab(\dims^2-1)}\Paren{\Tr[M]^2\Paren{\frac{\dims^2}{\ab}-1}+\Tr[M^2]\dims\Paren{1-\frac1{\ab}}}.
    \]
\end{lemma}
\begin{proof}
    Let $D_1$ be a diagonal matrix where the first $\PiRank$ diagonal entries are $1$ and 0 everywhere else, then $\Pi_1=UD_1U^\dagger$. Thus, we can apply~\cref{lem:haar-trace} and~\eqref{equ:wg-deg2},
    \begin{align*}
        &\expectDistrOf{U}{\Tr[\Pi_1 M]^2}\\
        =&\expectDistrOf{U}{\Tr[UD_1U^\dagger M]^2}=\expectDistrOf{U}{\Tr[D_1U^\dagger MU]^2}\\
        =&\sum_{\pi, \tau\in \Sim_2}\Wg_\dims(\pi^{-1}\tau)\permProd{D_1}{\pi}\permProd{M}{\tau}\\
        =&\frac{1}{\dims^2-1}(\Tr[D_1]^2\Tr[M]^2+\Tr[D_1^2]\Tr[M^2])-\frac{1}{\dims(\dims^2-1)}(\Tr[D_1]^2\Tr[M^2]+\Tr[D_1^2]\Tr[M]^2)\\
        =&\frac{\Tr[D_1]}{\dims(\dims^2-1)}\Paren{\Tr[M]^2(\dims \Tr[D_1]-1)+\Tr[M^2](\dims-\Tr[D_1])}\\
        =&\frac{1}{\ab(\dims^2-1)}\Paren{\Tr[M]^2\Paren{\frac{\dims^2}{\ab}-1}+\Tr[M^2]\dims\Paren{1-\frac1{\ab}}}.
    \end{align*}
    In the last two steps, we used $\Tr[D_1^2]=\Tr[D_1]=\dims/\ab$.
\end{proof}
\subsubsection{Upper bounding ${\normtwo{\p_\rho}}$}
Note that $\Tr[\rho^2]\le \Tr[\rho]=1$. From~\cref{lem:haar-square}, we can immediately upper bound $\expectDistrOf{U}{\normtwo{\p_\rho}^2}$,
\begin{align*}
    \expectDistrOf{U}{\normtwo{\p_\rho}^2}&=\frac{1}{\dims^2-1}\Paren{\Tr[\rho]^2\Paren{\frac{\dims^2}{\ab}-1}+\Tr[\rho^2]\dims\Paren{1-\frac1{\ab}}}\\
    &\le \frac{1}{\dims^2-1}\Paren{\frac{\dims^2}{\ab}-1+\dims-\frac{\dims}{\ab}}\\
    &=\frac{\dims-1}{\dims^2-1}\Paren{\frac{\dims}{\ab}+1}\\
    &=\frac{\dims+\ab}{\ab(\dims+1)}\le \frac{2}{\ab}.
\end{align*}
The final inequality is due to $\ab\le \dims$ and thus $\dims+\ab\le 2(\dims+1)$. Therefore by Markov's inequality, 
\begin{equation}
    \probaDistrOf{U}{\normtwo{\p_\rho}\ge \frac{10}{\sqrt{\ab}}}=\probaDistrOf{U}{\normtwo{\p_\rho}^2\ge \frac{100}{\ab}}\le \frac{\expectDistrOf{U}{\normtwo{\p_\rho}^2}}{100/\ab}\le \frac{1}{50}=0.02.
    \label{equ:high-prob-distr-l2-norm}
\end{equation}

\subsubsection{Lower bounding $\normtwo{\p_\rho-\p_{\qkn}}$ }
For convenience we set $Z={\normtwo{\p_\rho-\p_{\qkn}}^2}$. First from~\cref{lem:haar-square}, 
\begin{align}
    \expect{Z}&=\expectDistrOf{U}{\normtwo{\p_\rho-\p_{\qkn}}^2}=\frac{1}{\dims^2-1}\Paren{\Tr[\Delta]^2\Paren{\frac{\dims^2}{\ab}-1}+\Tr[\Delta^2]\dims\Paren{1-\frac1{\ab}}}\nonumber\\
    &=\Tr[\Delta^2]\cdot\frac{\dims}{\dims^2-1}\Paren{1-\frac{1}{\ab}}\nonumber\\
    &=\frac{\Tr[\Delta^2]}{\dims}\Paren{1-\frac{1}{\ab}}(1+O(\dims^{-2})).\label{equ:Z-mean}
\end{align}
Thus $\expect{Z}=\bigTheta{\frac{\Tr[\Delta^2]}{\dims}}$. 
Intuitively, the random variable should concentrate around its mean, so we expect $Z<\theta \expect{Z}$ with high probability. To formally prove this, we need to upper bound the second moment $\expect{Z^2}$ by roughly $\expect{Z}^2=\bigTheta{\frac{\Tr[\Delta^2]^2}{\dims^2}}$. 
\begin{align}
    \expect{Z^2}&=\sum_{x,y=1}^{\ab}\Tr[\Pi_x\Delta]^2\Tr[\Pi_y\Delta]^2\nonumber\\
    &=\ab\expectDistrOf{U}{\Tr[\Pi_1\Delta]^4}+\ab(\ab-1)\expectDistrOf{U}{\Tr[\Pi_1\Delta]^2\Tr[\Pi_2\Delta]^2}.\label{equ:Z-square}
\end{align}
The two terms are bounded in~\cref{lem:l2-order4,lem:l2-order4-cross}. The proofs are very technical and can be found in~\cref{app:lem:l2-order4,app:lem:l2-order4-cross} respectively.
\begin{lemma}
    \label{lem:l2-order4}
    Let $\Delta\in \C^{\dims\times\dims}$ be $\Tr[\Delta]=0$ and $\Pi=\Pi_1$ defined in~\cref{lem:haar-square}. Then,
    \[
    \expectDistrOf{U}{\Tr[\Pi\Delta]^4}\le \frac{3\Tr[\Delta^2]^2}{\dims^2\ab^2}\Paren{\Paren{1-\frac{1}{\ab}}^2+\frac{2\ab}{\dims}+O(\dims^{-1})}.
    \]
\end{lemma}
\begin{lemma}
\label{lem:l2-order4-cross}
    Let $\Delta\in \C^{\dims\times\dims}$ be trace-0 and $\Pi_1,\Pi_2$ defined in the beginning of~\cref{app:proof:12-norm-haar}. Then
    \[
    \expectDistrOf{U}{\Tr[\Pi_1\Delta]^2\Tr[\Pi_2\Delta]^2}\le \frac{\Tr[\Delta^2]^2}{\dims^2\ab^2}\Paren{1-\frac{2}{\ab}+\frac{3}{\ab^2}+O(\dims^{-1})}.
    \]
\end{lemma}

Combining~\eqref{equ:Z-square} and \cref{lem:l2-order4,lem:l2-order4-cross},
\begin{align*}
    \expect{Z^2}&=\frac{3\Tr[\Delta^2]^2}{\dims^2\ab}\Paren{\Paren{1-\frac{1}{\ab}}^2+\frac{2\ab}{\dims}+O(\dims^{-1})}+\frac{\Tr[\Delta^2]^2(\ab-1)}{\dims^2\ab}\Paren{1-\frac{2}{\ab}+\frac{3}{\ab^2}+O(\dims^{-1})}\\
    &=\frac{\Tr[\Delta^2]^2}{\dims^2}\Paren{\Paren{1-\frac{1}{\ab}}\Paren{\frac{3}{\ab}-\frac{3}{\ab^2}+1-\frac{2}{\ab}+\frac{3}{\ab^2}}+O(\dims^{-1}) }\\
    &=\frac{\Tr[\Delta^2]^2}{\dims^2}\Paren{1-\frac{1}{\ab^2}+O(\dims^{-1})}.
\end{align*}

Combining with~\eqref{equ:Z-mean}, by Paley-Zygmund, for $\theta\in(0, 1)$,
\begin{align*}
    \probaOf{Z>\theta \expect{Z}}&\ge (1-\theta)^2\frac{\expect{Z}^2}{\expect{Z^2}}\\
    &\ge(1-\theta)^2\frac{(1-\ab^{-1})^2}{1-\ab^{-2}}\cdot\frac{1+O(\dims^{-1})}{1+O(\dims^{-1})}\\
    &=(1-\theta)^2\frac{\ab-1}{\ab+1}\cdot\frac{1+O(\dims^{-1})}{1+O(\dims^{-1})}.
\end{align*}

Finally note that $\expect{Z}\simeq\frac{\Tr[\Delta^2]}{\dims}(1-1/\ab)\ge \frac{\hsnorm{\rho-\qkn}^2}{2\dims}$. For $\ab\ge 2$ and sufficiently large $\dims$, e.g. $\dims\ge 100$, 
\begin{align*}
    \probaOf{\normtwo{\p_\rho-\p_{\qkn}}>0.07\cdot \frac{\hsnorm{\rho-\qkn}}{\sqrt{\dims}}}&= \probaOf{Z>0.0049\frac{\hsnorm{\rho-\qkn}^2}{\dims}}\\
    &\ge \probaOf{Z>0.01\expect{Z}}\\
    &\ge 0.99^2\cdot\frac{1}{3}\cdot\frac{\ab-1}{\ab+1}\cdot\frac{1+O(\dims^{-1})}{1+O(\dims^{-1})}\ge 0.13.
\end{align*}
This completes the proof of \cref{lem:l2-norm-haar}.

\subsection{Upper bounding the 4th order terms}
\subsubsection{Proof of~\cref{lem:l2-order4}}
\label{app:lem:l2-order4}
\begin{proof}
    We directly apply~\cref{lem:haar-square}. Recall $\Pi=UDU^\dagger$ where $D=D_1$.
    \begin{align*}
        \expectDistrOf{U}{\Tr[\Pi\Delta]^4}&=\expectDistrOf{U}{\Tr[UDU^\dagger\Delta]^4}\\
        &=\sum_{\pi,\tau\in\Sim_4}\Wg_\dims(\pi^{-1}\tau)\permProd{D}{\pi}\permProd{\Delta}{\tau}
    \end{align*}
    We first compute $\permProd{D}{\pi}, \permProd{\Delta}{\tau}$ using~\eqref{equ:perm-prod}.  Since $D^2=D$ and $\Tr[D]=\PiRank=\dims/\ab$,
    \begin{equation}
    \label{equ:permprod-D}
        \permProd{D}{\pi}=\Tr[D]^{|\cycle(\pi)|}=(\dims/\ab)^{|\cycle(\pi)|}.
    \end{equation}
    Since $\Tr[\Delta]=0$, $\permProd{\Delta}{\tau}=0$ if $\tau$ has a cycle of length 1. Thus, $\permProd{\Delta}{\tau}\ne 0$ only if  $(4)$ or $(2^2)$ where 
    \begin{equation}
    \label{equ:permprod-Delta}
        \permProd{\Delta}{\tau}=\begin{cases}
            \Tr[\Delta^4],&\tau\in\{(4)\}\\
            \Tr[\Delta^2]^2, &\tau\in\{(2^2)\}.
        \end{cases}
    \end{equation}
    
    Therefore,
    \begin{align*}
        \expectDistrOf{U}{\Tr[\Pi_1\Delta]^4}&=\Tr[\Delta^2]^2\sum_{\pi\in\Sim_4, \tau\in(2^2)}\Wg_\dims(\pi^{-1}\tau)\permProd{D}{\pi}+\Tr[\Delta^4]\sum_{\pi\in\Sim_4, \tau\in(4)}\Wg_\dims(\pi^{-1}\tau)\permProd{D}{\pi}\\
        &=3\Tr[\Delta^2]^2\sum_{\pi\in\Sim_4}\Wg_\dims(\pi^{-1}\tau_{0})\permProd{D}{\pi}+6\Tr[\Delta^4]\sum_{\pi\in\Sim_4}\Wg_\dims(\pi^{-1}\tau_{1})\permProd{D}{\pi},
    \end{align*}
    where $\tau_0$ is a permutation with cycle type $(2^2)$, e.g. $(2,1,4,3)$ and $\tau_1$ has cycle type $(4)$, e.g. $(2,3,4,1)$. The second step follows by symmetry.

    The next step is to find all terms of $Wg_\dims(\pi^{-1}\tau)\permProd{D}{\pi}$ with order $\dims^{-2}$. Given $\tau$, the cycle types of $\pi^{-1}\tau$ and their multiplicities are listed in~\cref{tab:cycle-types-pi-tau}.
    \begin{table}[h]
        \centering
        \def\arraystretch{1.4}
        \begin{tabular}{|c|c|c|c|c|c|}
        \hline
            $\pi^{-1}$ & $(1^4)$ & $(21^2)$ &$(2^2)$ & $(31)$ &$(4)$   \\ \hline
           \multirow{2}{*}{$\pi^{-1}\tau, \tau\in(2^2)$}  &\multirow{2}{*}{ \color{red}$(2^2),1$} &{\color{red} $(21^2), 2$} &{\color{red}$(1^4),1$} &\multirow{2}{*}{$(31),8$} & $(21^2),4$ \\ 
           & &$(4),4$ & $(2^2),2$ & & $(4),2$\\\hline
           \multirow{3}{*}{$\pi^{-1}\tau, \tau\in(4)$} & \multirow{3}{*}{$(4),1$}  & $(2^2),2$ &$(21^2),2$ &$(4),4$ &$(31),4$\\
           & & $(31),4$ & $(4),1$ & $(21^2),4$ &{\color{blue}$(1^4),1$}\\
           & & & & & $(2^2),1$\\\hline
        \end{tabular}
        \caption{Cycle types of $\pi^{-1}\tau$ when $\tau$ has cycle type $(2^2)$ and $(4)$. The number after each cycle type indicates the multiplicity for fixed $\tau$. The red entries are of order $\Theta_{\ab}(\dims^{-2})$ and the blue entry is $\dims^{-3}\ab^{-1}$. All other terms are $O(\dims^{-3}\ab^{-2})$.}
        \label{tab:cycle-types-pi-tau}
    \end{table}
    
     Using~\eqref{equ:permprod-D},~\cref{tab:wg-deg-4}, and that $\Tr[\Delta^4]\le \Tr[\Delta^2]^2$, we have
     \begin{align*}
         \expectDistrOf{U}{\Tr[\Pi_1\Delta]^4}&=3\Tr[\Delta^2]^2\Paren{\frac{1}{\dims^2\ab^4}-\frac{2}{\dims^2\ab^3}+\frac{1}{\dims^2\ab^2}+O(\dims^{-3})}+6\Tr[\Delta^4]\Paren{\frac{1}{\dims^3\ab}+\bigO{\dims^{-3}\ab^{-2}}}\\
         &\le \frac{3\Tr[\Delta^2]^2}{\dims^2\ab^2}\Paren{\Paren{1-\frac{1}{\ab}}^2+\frac{2\ab}{\dims}+O(\dims^{-1})}.
     \end{align*}
     The proof is complete.
\end{proof}

\subsubsection{Proof of~\cref{lem:l2-order4-cross}}
\label{app:lem:l2-order4-cross}
\begin{proof}
    Recall that $\Pi_x=UD_xU^\dagger$. First, we apply~\cref{lem:haar-twirl}.
    \begin{align*}
        \expectDistrOf{U}{\Tr[\Pi_1\Delta]^2\Tr[\Pi_2\Delta]^2}&=\expectDistrOf{U}{\Tr\left[(\Pi_1^{\otimes 2}\otimes\Pi_2^{\otimes2})\Delta^{\otimes 4}\right]}\\
        &=\Tr\left[\expectDistrOf{U}{U^{\otimes4}(D_1^{\otimes2}\otimes D_2^{\otimes2})(U^{\dagger})^{\otimes4}\Delta^{\otimes4}}\right]\\
        &=\sum_{\pi,\tau\in\Sim_4}\Wg_\dims(\pi^{-1}\tau)\Tr[P_{\pi}(D_1^{\otimes2}\otimes D_2^{\otimes2})]\Tr[P_{\tau}\Delta^{\otimes4}]\\
        &=\sum_{\pi,\tau\in\Sim_4}\Wg_\dims(\pi^{-1}\tau)\Tr[P_{\pi}(D_1^{\otimes2}\otimes D_2^{\otimes2})]\permProd{\Delta}{\tau}.
    \end{align*}
    We need to compute $\Tr[P_{\pi}(D_1^{\otimes2}\otimes D_2^{\otimes2})]$. Write $\D_1=\sum_{i=1}^{\PiRank}\qproj{e_i}$ and $\D_2=\sum_{i=1}^{\PiRank}\qproj{f_i}$ where $\qdotprod{e_i}{f_j}=0$. Then,
    \[
    D_1^{\otimes2}\otimes D_2^{\otimes2}=\sum_{i,j,k,l}\qproj{e_i,e_j,f_k,f_l}.
    \]
    Let $\Sim^*=\Sim_2\times\Sim_2$ be the set of permutations that maps $\{1, 2\}$ to $\{1, 2\}$ and $\{3,4\}$ to $\{3,4\}$, which only has 4 permutations,
    \[
    \Sim^*=\{(1,2,3,4),(2,1,3,4),(1,2,4,3),(2,1,4,3)\}.
    \]
    If $\pi\notin\Sim^*$, then $\Tr[P_{\pi}(D_1^{\otimes2}\otimes D_2^{\otimes2})]=0$. For example, if 1 maps to 3, then
    \begin{align*}   \Tr[P_\sigma\qproj{e_i,e_j,f_k,f_l}]&=\Tr[\qbit{\cdot,\cdot,e_i,\cdot}\qadjoint{e_i,e_j,f_k,f_l}]\\
    &=\qdotprod{e_i}{\cdot}\qdotprod{e_j}{\cdot}\qdotprod{f_k}{e_i}\qdotprod{f_l}{\cdot}\\
    &=0.
    \end{align*}
    Thus we only need to consider $\pi\in \Sim^*$. Observe that for $\pi\in\Sim^*$, we can write $P_\pi=P_{\pi_1}\otimes P_{\pi_2}$ where $\pi_1, \pi_2\in \Sim_2$. Therefore,
    \begin{align*}
        \Tr[P_{\pi}(D_1^{\otimes2}\otimes D_2^{\otimes2})&=\Tr[P_{\pi_1}\otimes P_{\pi_2}(D_1^{\otimes2}\otimes D_2^{\otimes2})\\
        &=\Tr[P_{\pi_1}D_1^{\otimes2}]\Tr[P_{\pi_2}D_2^{\otimes2}]\\
        &=\permProd{D_1}{\pi_1}\permProd{D_2}{\pi_2}\\
        &=\Paren{\frac{\dims}{\ab}}^{|\cycle(\pi)|}.
    \end{align*}
    We proceed to evaluate the expression,
    \begin{align*}
        &\expectDistrOf{U}{\Tr[\Pi_1\Delta]^2\Tr[\Pi_2\Delta]^2}\\
        =&\sum_{\pi\in \Sim^*,\tau\in\Sim_4}\Wg_{\dims}(\pi^{-1}\tau)\Paren{\frac{\dims}{\ab}}^{|\cycle(\pi)|}\permProd{\Delta}{\tau}\\
        =&\Tr[\Delta^4]\sum_{\pi\in \Sim^*,\tau\in(4)}\Wg_{\dims}(\pi^{-1}\tau)\Paren{\frac{\dims}{\ab}}^{|\cycle(\pi)|}+\Tr[\Delta^2]^2\sum_{\pi\in\Sim^*,\tau\in(2^2)}\Wg_{\dims}(\pi^{-1}\tau)\Paren{\frac{\dims}{\ab}}^{|\cycle(\pi)|}.
    \end{align*}
    The cycle types in~\cref{tab:cycle-types-pi-tau} also apply here. Only for $\tau\in(2^2)$, the term has order $\Theta_{\ab}(\dims^{-2})$, and for all other terms the order is $O(\dims^{-3}\ab^{-2})$. We can count the multiplicity of the red terms.
    \begin{enumerate}
        \item $\Wg_{\dims}(2^2)$ has 3 terms, where $\pi=e$, and $\tau\in(2^2)$ can be arbitrary.
        \item $\Wg_{\dims}(21^2)$ has 2 terms. From~\cref{tab:cycle-types-pi-tau}, we know that $|\{\pi^{-1}\tau\in(21^2):\pi\in(21^2),\tau\in(2^2)\}|=2\times 3=6$. By symmetry for each fixed $\pi_0\in(21^2)$, $|\{\pi_0^{-1}\tau\in(21^2):\tau\in(2^2)\}|=6/|(21^2)|=1$. There are two permutations $\pi\in\Sim^*$ with cycle $(21^2)$, so the multiplicity is 2.
        \item $\Wg_\dims(1^4)$ has 1 term. The only $\pi\in\Sim^*$ with cycle $(2^2)$ is $\pi=(2,1,4,3)$, and $\pi^{-1}\tau=e$ if and only if $\tau=\pi=(2,1,4,3)$.
    \end{enumerate}
    Summarizing the results and using the asymptotics in~\cref{tab:wg-deg-4}, 
    \begin{align*}
        \expectDistrOf{U}{\Tr[\Pi_1\Delta]^2\Tr[\Pi_2\Delta]^2}&=\Tr[\Delta^4]O(\dims^{-3}\ab^{-2})+\Tr[\Delta^2]^2\Paren{\frac{1}{\dims^2\ab^2}-\frac{2}{\dims^2\ab^3}+\frac{3}{\dims^2\ab^4}}\\
        &\le \frac{\Tr[\Delta^2]^2}{\dims^2\ab^2}\Paren{1-\frac{2}{\ab}+\frac{3}{\ab^2}+O(\dims^{-1})}.
    \end{align*}
    The proof is complete.
\end{proof}

\end{document}